\documentclass[letterpaper,11pt]{article}
\usepackage[utf8]{inputenc}
\usepackage{amssymb,amsmath,amsthm,mathtools,amsbsy}
\usepackage{subcaption}
\usepackage{graphicx}

\usepackage[margin=1.15in]{geometry}
\newcommand{\p}{\partial}

\newcommand{\h}{\theta}

\newcommand{\hi}{\theta_i}

\newcommand{\bfx}{\mathbf{x}}
\newcommand{\bfy}{\mathbf{y}}
\newcommand{\bfX}{\boldsymbol{\alpha}}

\newcommand{\ei}{\mathbf{e}_1}
\newcommand{\ej}{\mathbf{e}_2}
\newcommand{\ek}{\mathbf{e}_3}

\newcommand{\fl}{\Phi}

\newcommand{\R}{\mathbb{R}}
\def\*#1{\mathbf{#1}}

\theoremstyle{plain}
\newtheorem{thm}{Theorem}[section]

\theoremstyle{remark}

\newtheorem{rmk}[thm]{Remark}
\theoremstyle{definition}

\newcommand*\samethanks[1][\value{footnote}]{\footnotemark[#1]}

\begin{document}

\title{Self-organization on Riemannian manifolds} 

\author{Razvan C. Fetecau \thanks{Department of Mathematics, Simon Fraser University, 8888 University Dr., Burnaby, BC V5A 1S6, Canada}
\and Beril Zhang \samethanks}

\maketitle

\begin{abstract}
We consider an aggregation model that consists of an active transport equation for the macroscopic population density, where the velocity has a nonlocal functional 
dependence on the density, modelled via an interaction potential. We set up the model on general Riemannian manifolds and provide a framework for constructing interaction potentials which lead to equilibria that are constant on their supports. We consider such potentials for two specific cases (the two-dimensional sphere and the two-dimensional hyperbolic space) and investigate analytically and numerically the long-time behaviour and equilibrium solutions of the aggregation model on these manifolds. Equilibria obtained numerically with other interaction potentials are also presented.

\end{abstract}



\section{Introduction}
\label{sect:intro}

The literature on self-collective behaviour of autonomous agents (e.g., biological organisms, robots, nanoparticles, etc) has been growing very fast recently. One of the main interests of such research is to understand how swarming and flocking behaviours emerge in groups with no leader or external coordination. Such behaviours occur for instance in natural swarms, e.g., flocks of birds or schools of fish \cite{Camazine_etal, Couzin_etal}. Also, swarming and flocking of artificial mobile agents (e.g., robots) in the absence of a centralized coordination mechanism is of major interest in engineering \cite{JaLiMo2003, JiEgerstedt2007}. Consequently, there exists a variety of models for swarming or flocking, ranging from difference equations (discrete in both time and space) \cite{Czirok_PRL1999, Couzin_etal, Cucker:Smale} to ordinary/partial differential equations (continuous in time and discrete/continuous in space) \cite{Chuang_etal, M&K, Eftimie2}.

In this paper we consider an aggregation model that consists in an integro-differential equation for the evolution of a population density $\rho(\bfx,t)$ on a Riemannian manifold $M$:
\begin{subequations}
\label{eqn:model}
\begin{gather}
\rho_{t}+\nabla_M \cdot(\rho \mathbf{v})=0, \label{eqn:pde} \\
\mathbf{v}=-\nabla_M K\ast\rho. \label{eqn:v}%
\end{gather}
\end{subequations}
Here, $K:M\times M \to \R$ is an interaction potential, which models social interactions such as attraction and repulsion, and $K \ast \rho$ is defined as:
\begin{equation}
\label{eqn:conv}
K * \rho(\bfx) = \int_M K(\bfx,\bfy) \rho(\bfy)  d \mu (\bfy),
\end{equation}
where the integration is with respect to the canonical volume form $\mu$ of the Riemannian manifold $M$. 

Model \eqref{eqn:model} set up in Euclidean space $\R^n$ has received a great deal of interest in recent years. On one hand, it has been used in numerous applications  such as swarming in biological groups \cite{M&K}, material science and granular media \cite{CaMcVi2006}, self-assembly of nanoparticles \cite{HoPu2005}, opinion formation \cite{MotschTadmor2014},  robotics and space missions \cite{JiEgerstedt2007}, and molecular dynamics \cite{Haile1992}. On the other hand, there have been excellent progress and insight in the numerics and analysis for model \eqref{eqn:model}. It has been shown numerically that the model can capture a wide variety of self-collective or swarm behaviours, such as aggregations on disks, annuli, rings, soccer balls, etc \cite{KoSuUmBe2011, Brecht_etal2011, BrechtUminsky2012}.  At the same time, the mathematical analysis of model \eqref{eqn:model} in $\R^n$ has posed challenges that stimulated a rich and diverse literature. The issues addressed include the well-posedness of the initial-value problem  \cite{BodnarVelasquez2, BertozziLaurent, Figalli_etal2011, BeLaRo2011}, the long time behaviour of its solutions \cite{Burger:DiFrancesco, LeToBe2009, FeRa10, BertozziCarilloLaurent, FeHuKo11,FeHu13}, evolution in domains with boundaries \cite{CarrilloSlepcevWu2016, FeKo2017}, and studies on minimizers for the associated interaction energy \cite{BaCaLaRa2013, Balague_etalARMA, ChFeTo2015}.

Despite the extensive research on equation \eqref{eqn:model} in Euclidean spaces, there is very little done for the aggregation model posed on arbitrary surfaces or manifolds.
In \cite{WuSlepcev2015}, the authors investigate the well-posedness of the aggregation model \eqref{eqn:model} on Riemannian manifolds, but in a certain restrictive setting (as detailed in the next paragraph). To the best of our knowledge, the present paper is the first to provide a formulation of the aggregation model on general Riemannian manifolds, which we believe has important applications (e.g, in robotics). We point out that it has been only very recently that other classes of models have been considered on surfaces and manifolds too; see for instance \cite{Li2015} for a Vicsek-type model \cite{Czirok_PRL1999} set up on a sphere.

With very few exceptions, in the studies of the aggregation model in $\R^n$, the interaction potential $K(\bfx,\bfy)$ at two locations $\bfx$ and $\bfy$ is assumed to depend on the Euclidean distance between the two points, i.e., $K(\bfx,\bfy) = K(|\bfx- \bfy|)$. In particular, interactions are {\em symmetric}, as two individuals positioned at $\bfx$ and $\bfy$ sense each other equally. Note that for such interaction potentials, the operation $\ast$ defined in \eqref{eqn:conv} is a standard convolution between two scalar functions on $\R^n$. As alluded to above, in \cite{WuSlepcev2015}, the aggregation model is indeed posed on a Riemannian manifold, but the setup is very restrictive, as it is assumed there that the manifold is a subset of $\R^n$ and that mutual interactions depend on the {\em Euclidean} distance (in $\R^n$) between points. 

Different from \cite{WuSlepcev2015}, we consider in this paper general Riemannian manifolds $(M,g)$, for which interactions depend on the {\em geodesic} distance (on $M$) between points. Mathematically, by an abuse of notation, we write:
\begin{equation}
\label{eqn:K-dist}
K(\bfx,\bfy) = K(d(\bfx,\bfy)),
\end{equation}
where $d(\bfx,\bfy)$ denotes the geodesic distance between $\bfx$ and $\bfy$. Not only that \eqref{eqn:K-dist} generalizes naturally the Euclidean setup, but we also find it more meaningful. For example, consider applications of the model in engineering (robotics) \cite{Gazi:Passino, JiEgerstedt2007}, where individual agents/robots are restricted by environment or mobility constraints to remain on a certain manifold. To achieve efficient swarming or flocking, agents must approach each other along geodesics, so the geodesic distance should be built into the model, which in this class of models amounts to incorporating it into the interaction potential. 

Similar to the setup in Euclidean spaces, we assume that there are no limitations on the mutual sensing of individuals, that is, all individuals sense each other. For the example above (coordination of mobile agents), the assumption is valid provided the agents are set to communicate globally with each other via a central unit. For coordination of biological agents, one can assume that individuals possess a sensing mechanism (such as smell for instance) which enables them to communicate with each other regardless of the geometry of the space they live in. We note here that from a mathematical point of view, including local sensing and/or asymmetry/anisotropy in the model is expected to bring up major challenges (see for instance \cite{EvFeRy2015} for a study of anisotropic interactions in model \eqref{eqn:model} posed on $\R^2$).

Model \eqref{eqn:model} is in the form of a continuity equation for the density $\rho$. Note that this is an {\em active} transport equation, as the velocity field defined in \eqref{eqn:v} depends on $\rho$. The interpretation of \eqref{eqn:model} as an aggregation model is in fact encoded in \eqref{eqn:v}: by interacting with a point mass at location $\bfy$, the point mass at $\bfx$ moves either towards or away from $\bfy$. The velocity $\mathbf{v}$ at $\bfx$ computed by \eqref{eqn:v} takes into account all contributions from interactions with point masses $\bfy \in M$. Also, in geometric terms, the continuity equation \eqref{eqn:pde} represents the transport of the volume form $\rho \mu$ along the flow on $M$ generated by the tangent vector field $\mathbf{v}$ \cite[Chapter 8]{AMR1988}. Equivalently, the mass in each subregion of the manifold remains constant through time, as the subregion evolves by the flow. In particular, the total mass $m = \int_M \rho \, d \mu$ is conserved. This geometric interpretation of the aggregation equation will play a major role in the paper.

In the present research we investigate solutions for model \eqref{eqn:model} posed on two simple manifolds: the two-dimensional sphere and the two-dimensional hyperbolic space. We design interaction potentials for which the equilibrium densities have a simple structure that is amenable to analytic investigations (specifically, equilibria are constant on their supports). The strategy for designing such potentials is inspired by \cite{FeHuKo11, FeHu13}, where similar goals were pursued for the aggregation model in Euclidean spaces. Our work is the first to demonstrate that model \eqref{eqn:model} set up on manifolds leads to swarming behaviour that can be studied mathematically. We also perform numerical investigations for other interaction potentials and find a diverse set of equilibrium solutions. Given the outstanding interest shown recently in the aggregation model on $\R^n$, we hope to have set up the stage for further studies on swarming and self-collective behaviour on general manifolds, opening new perspectives and motivating applied mathematicians to expand the research on this class of models to novel applications.  

We finally note that the continuum model \eqref{eqn:model} has an immediate discrete/ODE analogue, in which one considers the time evolution of a fixed number of individuals/particles on a manifold. While the discrete model is used in this paper only for numerical purposes (Sections \ref{subsubsect:numerics-sphere} and \ref{subsubsect:numerics-hyper}), it has an interest in its own (e.g., for applications in biology or robotics). Specifically, equilibria of constant densities, as achieved in this paper, translate in the discrete setup to agents covering uniformly (with respect to the metric) a certain space modelled as a Riemannian manifold. This relates to the coverage problem in robotics, where the goal is to have a group of robots well-distributed over a region/area, so that it achieves an optimal coverage needed for surveillance/tracking (see \cite{CortesEgerstedt2017} for a recent survey on coordinated control of robot systems). The ODE model has been extensively employed for the aggregation model in Euclidean space, either as a numerical tool \cite{KoSuUmBe2011, Brecht_etal2011, Balague_etalARMA, FeHuKo11}, or in rigorous analysis studies that establish the passage from discrete to continuum by mean-field limits  \cite{CaChHa2014}. Consequently, we also expect a rich use and interest for the discrete aggregation model on Riemannian manifolds, as set up in the present paper.

The summary of the paper is as follows. In Section \ref{sect:prelim} we present some preliminaries and provide motivation to the present research. In Section \ref{sect:sphere} we set up the aggregation model on the two-dimensional sphere, with a certain choice of the interaction potential, and investigate the long-time behaviour of solutions. A similar study is done in Section \ref{sect:hyperboloid} for the model on the two-dimensional hyperbolic space. In Section \ref{sect:future} we showcase some numerical simulations with other interaction potentials to further motivate the aggregation model proposed here.



\section{Preliminaries and motivation}
\label{sect:prelim}

In this section we present briefly some standard concepts from differential geometry that will be used in the sequel, as well as bring motivation to the present work.

\paragraph{Local coordinates.} Consider a generic $n$-dimensional Riemannian manifold $(M,g)$ with local coordinates $(x^1, x^2, \dots, x^n)$ and local metric coefficients $g_{ij}$.

Expressed in the local basis $\left \{ \frac{\p }{\p x^1}, \frac{\p }{\p x^2}, \dots, \frac{\p }{\p x^n}\right \}$ of the tangent space, the gradient (with respect to the metric $g$) of a scalar function $f$ on $M$ is the tangent vector $\nabla_M f$ given by
 \begin{equation}
\label{eqn:nabla-M}
 \nabla_M f = g^{ji} \frac{\p f}{\p x^j} \frac{\p }{\p x^i},
 \end{equation}
where $g^{ij}$ are the entries of the metric's inverse, and we used the Einstein convention on index summation \cite{Kuhnel-book}.

Given a tangent vector field $\mathbf{F} = F^i \frac{\p }{\p x^i}$ on $M$, its divergence $\nabla_M \cdot \mathbf{F}$ is given in local coordinates by 
\begin{equation}
\label{eqn:div-M}
\nabla_M \cdot \mathbf{F} = \frac{1}{\sqrt{|g|}} \frac{\p} {\p x^i} \left( \sqrt{|g|} F^i \right),
\end{equation}
where $|g|$ denotes the determinant of the metric.

By combining \eqref{eqn:nabla-M} and \eqref{eqn:div-M} one then finds the Laplace-Beltrami operator $\Delta_M$ in local coordinates:
\begin{equation}
\label{eqn:LB-M}
\Delta_M f =  \frac{1}{\sqrt{|g|}} \frac{\p} {\p x^i} \left( \sqrt{|g|} \,  g^{ji} \frac{\p f}{\p x^j} \right).
\end{equation}

Aside from the operators above, we will also use the representation in local coordinates of the canonical volume form $\mu$ on $(M,g)$ \cite{AMR1988}:
\begin{equation}
\label{eqn:vol-form}
\mu = \sqrt{|g|} \, dx^1 \wedge  \dots \wedge dx^n.
\end{equation}


\paragraph{Continuity equation on Riemannian manifolds.} As discussed in \cite[Chapter 8.2]{AMR1988}, the continuity equation \eqref{eqn:model} posed on Riemannian manifolds has the following interpretation. 

Consider the flow $\fl_t : M \to M$ on $M$ generated by the vector field $\mathbf{v}$, that is:
\begin{equation}
\label{eqn:charact}
\frac{d \Phi_t(\bfX)}{dt} = \mathbf{v} (\Phi_t(\bfX),t).
\end{equation}
For a fixed time $t$, denote by $\rho_t(\bfx) = \rho(\bfx,t)$, and also recall that $\mu$ represents the canonical volume form on $M$. Then, the continuity equation \eqref{eqn:model} is equivalent to the transport of the volume form $\rho_t \mu$ along the flow $\Phi_t$, that is:
\[
\frac{d}{dt} \Phi_t^\ast(\rho_t \mu) =0, \qquad \text{ or } \qquad \Phi_t^\ast(\rho_t \mu) = \rho_0 \mu.
\]
Here, $\Phi_t^\ast$ denotes the pull-back by $\Phi_t$. Moreover, this is equivalent to:
\begin{equation}
\label{eqn:continuity-alt}
\Phi_t^\ast(\rho_t) \cdot J(\Phi_t) = \rho_0,
\end{equation}
where $J(\Phi_t)$ denotes the Jacobian of $\Phi_t$ (with respect to the canonical volume form $\mu$).

We also note that for a smooth, invertible flow map $\Phi_t$, the Jacobian $J(\Phi_t)$ satisfies (see \cite[Proposition 7.1.10]{AMR1988}):
\begin{equation}
\label{eqn:J-flow}
\Phi_t ^\ast \mu = J(\Phi_t)\,  \mu.
\end{equation}


\paragraph{Other properties of the model.} {\em Energy.} The aggregation model \eqref{eqn:model} in Euclidean spaces can be formulated as a gradient flow on the space of probability measures with finite second moments, equipped with the $2$-Wasserstein metric \cite{AGS2005}.  Such interpretation exists as well for the model on Riemannian manifolds with {\em Euclidean} pairwise interactions from \cite{WuSlepcev2015}. The general model \eqref{eqn:model} also has an energy associated to it, which decays with time. Indeed, define:
\begin{equation}
\label{eqn:energy}
E[\rho] = \frac{1}{2} \iint_{M \times M}  K(\bfx,\bfy) \rho(\bfx) \rho(\bfy) \, d \mu(\bfx) d \mu(\bfy).
\end{equation}

By formally computing the evolution of the energy in time, we find:
\begin{align}
\frac{d E}{dt} &= -\int_M \nabla_M \cdot (\rho \mathbf{v}) \, K \ast \rho (\bfx)  d \mu(\bfx) \nonumber \\
& = - \int_M \rho(\bfx) \, g(\mathbf{v},\mathbf{v}) \, d \mu(\bfx) \leq 0. \label{eqn:energy-decay}
\end{align}
For the first equal sign in the derivation above we used the symmetry of the potential and equations \eqref{eqn:model} and \eqref{eqn:conv}. For the second equal sign we used the formula \cite[Proposition 6.5.17]{AMR1988}:
\[
 \nabla_M \cdot (\rho \mathbf{v}) \, K \ast \rho= \nabla_M \cdot (\rho \, K \ast \rho \, \mathbf{v}) - \nabla_{\rho \mathbf{v}}(K \ast \rho),
\]
where $\nabla_{\rho \mathbf{v}}(K \ast \rho)$ denotes the covariant derivative of the scalar function $K \ast \rho$  along the vector field $\rho \mathbf{v}$. Integrating the equation above over $M$, by divergence theorem (either assuming that $M$ has no boundary or that $M$ is non-compact, but density vanishes at infinity) the first term in the right-hand-side yields zero. Then, \eqref{eqn:energy-decay} follows from $\nabla_{\rho \mathbf{v}}(K \ast \rho) = \rho \nabla_{\mathbf{v}}(K \ast \rho)$ (linearity of the covariant derivative) and the definition of gradient on $M$ by which:
\[
\nabla_{\mathbf{v}}(K \ast \rho) = g(\underbrace{\nabla_M K\ast \rho}_{=-\mathbf{v}}, \mathbf{v}).
\]

Though we do not make any further energy considerations in this work, we believe that there is a rich potential for applications of the theory on gradient flows as developed in  \cite{AGS2005, Savare2007} to the model investigated here. In addition, the study of equilibria of model \eqref{eqn:model} as minimizers of energy $E$ seems a very interesting direction to pursue as well. This approach has been proven very successful for the model in $\R^n$ \cite{ChFeTo2015, SiSlTo2015, CaCaPa2015,Balague_etalARMA}.
\medskip

{\em Centre of mass.} An important property of model \eqref{eqn:model} in Euclidean space $\R^n$ is the conservation of the centre of mass. This can be derived easily as follows: multiply equation \eqref{eqn:model} by $\bfx$ and integrate over $\R^n$, then use integration by parts and the symmetry of $K$ to conclude that $\int_{\R^n} \rho(\bfx) d \bfx$ remains constant in time. 

In the context of Riemannian manifolds, as a generalization of the usual centre of mass in $\R^n$, we consider the $L^2$ Riemannian centre of mass (also known as the Karcher mean) \cite{Karcher1977, Afsari2011}. The $L^2$ Riemannian centre of mass (simply referred throughout as the Riemannian centre of mass) of a subset $A \subset M$ is a minimizer in $M$ of the function
\[
f(\bfx ) = \frac{1}{2} \int_A d^2(\bfx,\bfy) d\mu(\bfy).
\]
One can check indeed that for $M=\R^n$, $f$ has a unique minimizer which coincides with the usual centre of mass of set $A$. For general manifolds, existence and uniqueness of the Riemannian centre of mass, along with numerical methods for finding it, are delicate issues \cite{Afsari2011, AfsariTronVidal2013}.

We make the important observation here that model \eqref{eqn:model} does {\em not} necessarily conserve the Riemannian centre of mass. Explaining this fact in detail would be an unnecessary detour for the purpose of this paper. To give some intuition on why such result is not expected to hold in general however, we point our that in $\R^n$, by symmetry of $K$ (see \eqref{eqn:K-dist}, which in $\R^n$ it amounts to $K(\bfx,\bfy) = K(|\bfx-\bfy|)$), $\nabla K$ is antisymmetric in $\bfx$, $\bfy$, i.e., $\nabla_\bfx K(\bfx,\bfy) = - \nabla_\bfy K(\bfx,\bfy)$. This property is used in an essential way to show conservation of centre of mass in $\R^n$. On the other hand, such a property would not even make immediate sense on general Riemannian manifolds, as the manifold gradients with respect to $\bfx$ and $\bfy$ lie in tangent spaces at different points ($\nabla_{M,\bfx} K(\bfx,\bfy) \in T_\bfx M$ and  $\nabla_{M,\bfy} K(\bfx,\bfy) \in T_\bfy M$). The lack of conservation of centre of mass brings an additional challenge for analytical investigations of model \eqref{eqn:model}. We return to this point at the end of Section \ref{sect:hyperboloid}.



\paragraph{The aggregation model in Euclidean space.} To motivate and put the present work in context, we present first some key ideas and calculations for the aggregation model in Euclidean space, that is, model \eqref{eqn:model} set up on $M=\R^n$ endowed with the standard Euclidean metric. We refer here to the research in \cite{FeHuKo11, FeHu13}, where one of the main goals was to design interaction potentials for the aggregation model in $\R^n$ that yield equilibrium states which are biologically relevant and at the same time, are simple enough to be investigated analytically. 

Of particular importance in \cite{FeHuKo11} is an attractive-repulsive interaction potential that yields equilibria of constant densities and compact support. Specifically, this potential consists of Newtonian repulsion and quadratic attraction, which in two dimensions ($n=2$) it amounts to:
\begin{equation}
\label{eqn:K-plane}
K(\bfx,\bfx') = G(\bfx;\bfx') + \frac{1}{2}  |\bfx - \bfx'|^2,
\end{equation}
where $G(\bfx;\bfx')$ denotes the Green's function for the negative Laplacian in $\R^2$:
\begin{equation}
G(\bfx;\bfx') = -\frac{1}{2\pi} \log{|\bfx - \bfx'|}.
\label{eqn:G-plane}
\end{equation}
Note that
\begin{equation}
\label{eqn:KLaplacian-R2}
\Delta K(\bfx,\bfx') = -\delta_{\bfx'} + 2,
\end{equation}
and hence, from \eqref{eqn:v},
\begin{equation}
\begin{aligned}
 \nabla \cdot \mathbf{v} &= -\Delta K * \rho \\
 &= \rho - 2m,
\label{eqn:div-v-R2}
\end{aligned}
\end{equation}
where $m$ is the (constant) total mass. The key conclusion from \eqref{eqn:div-v-R2} is that with this choice of interaction potential, $\nabla \cdot \mathbf{v}$ is a {\em local} quantity, despite the fact that $\mathbf{v}$ itself is nonlocal, as given by \eqref{eqn:v}  through a convolution.

Expand now
\begin{equation}
\label{eqn:div-expand}
\nabla \cdot (\rho \mathbf{v}) = \mathbf{v} \cdot \nabla \rho + \rho \nabla \cdot \mathbf{v},
\end{equation}
and write the continuity equation \eqref{eqn:pde} as 
\begin{equation}
\label{eqn:pde-expand}
\rho_t + \mathbf{v} \cdot \nabla \rho = - \rho \nabla \cdot \mathbf{v}.
\end{equation}
The left-hand-side of \eqref{eqn:pde-expand} is the material derivative of the density $\rho$ along the flow generated by $\mathbf{v}$. Hence, by \eqref{eqn:pde-expand} and \eqref{eqn:div-v-R2}, along the particle path $\Phi_t(\bfX)$ that originates from location $\bfX$ (see equation \eqref{eqn:charact}), $\rho(\Phi_t(\bfX),t)$ satisfies
\begin{equation}
\label{eqn:evol-charact}
\frac{D}{Dt} \rho =-\rho (\rho - 2m).
\end{equation}
Note that the right-hand-side of \eqref{eqn:evol-charact} depends only on $\rho$ evaluated along the carrying characteristic $\Phi_t(\bfX)$, and not on values of $\rho$ along any other other characteristic. Therefore, this ODE for $\rho(\Phi_t(\bfX),t)$ can be investigated individually. 

From \eqref{eqn:evol-charact} one infers immediately that the solution $\rho(\Phi_t(\bfX),t)$ approaches the value $2m$ as $t \to \infty$. Consequently, equilibria of model \eqref{eqn:model} in $\R^2$, with potential \eqref{eqn:K-plane}, have constant densities on their supports. In fact, it was shown in \cite{FeHuKo11, BertozziLaurentLeger} that the constant density supported on a disk, i.e., 
\begin{equation}
\bar{\rho}(\bfx)=
\begin{cases}
2m & \quad\text{ if }\;|\bfx|<\frac{1}{\sqrt{2\pi}}\\
0 & \quad\text{ otherwise },
\end{cases}
\label{eqn:sstate-R2}
\end{equation}
is a global attractor for solutions of model \eqref{eqn:model} with potential \eqref{eqn:K-plane} in $\R^2$; see Figure \ref{fig:plane} for a numerical simulation using a particle method.

The considerations above extend to the aggregation model set in Euclidean $\R^n$ of arbitrary dimension. Indeed, one can construct an interaction potential that leads to similar long-time behaviour of solutions (for $G(\bfx,\bfx')$ one would have to use the fundamental solution of the negative Laplacian in $\R^n$ instead).

\begin{figure}[htb]
 \begin{center}
 \begin{tabular}{cc}
 \includegraphics[width=0.5\textwidth]{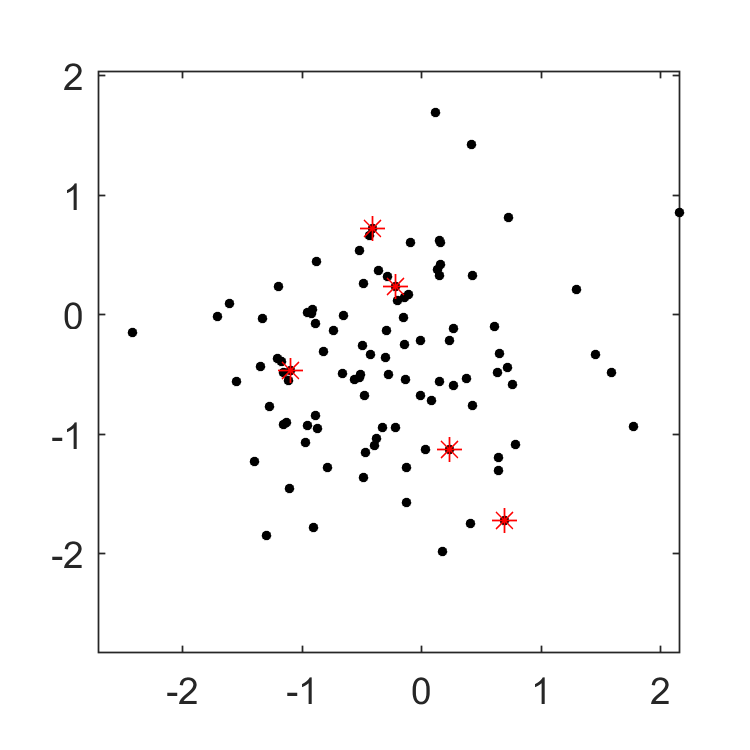} &
 \includegraphics[width=0.5\textwidth]{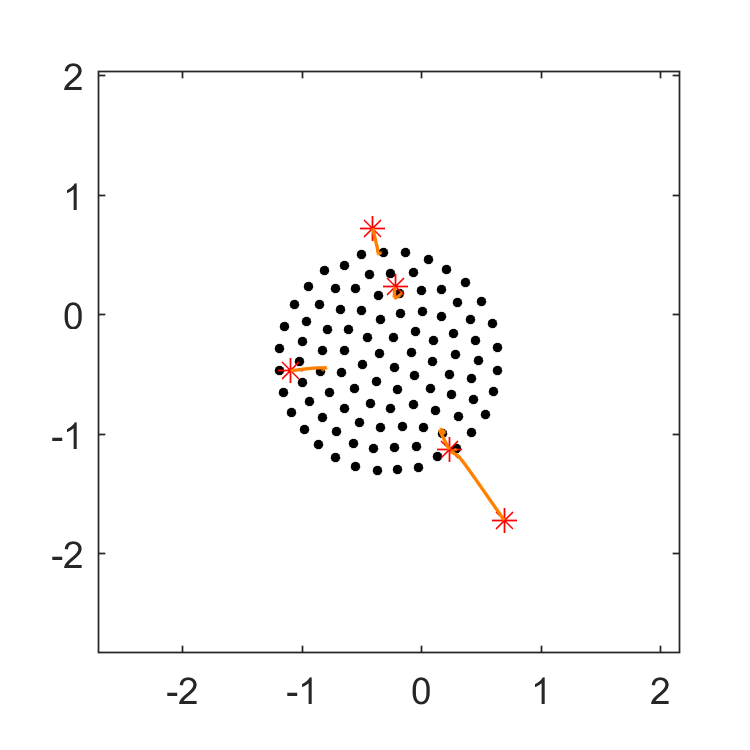} \\
 (a) & (b)
 \end{tabular}
\caption{Numerical simulation with $N=100$ particles for model \eqref{eqn:model} in the Euclidean plane with interaction potential given by \eqref{eqn:K-plane}-\eqref{eqn:G-plane}. (a) Random initial configuration of particles. (b) Starting from the configuration in (a), the model evolves into a uniform particle distribution supported on a disk -- see \eqref{eqn:sstate-R2}. The solid lines represent the trajectories of the particles indicated by stars in figure (a).}
\label{fig:plane}
\end{center}
\end{figure}


\paragraph{Motivation for the present work.} The ideas above can be extended to model \eqref{eqn:model} posed on arbitrary Riemannian manifolds $(M,g)$. Indeed, similar to \eqref{eqn:KLaplacian-R2}, consider an interaction potential that satisfies
\begin{equation}
\label{eqn:Mmotiv}
\Delta_M K(\bfx,\bfx') = -\delta_{\bfx'} + C,
\end{equation}
where $C>0$ a constant. Then, as in \eqref{eqn:div-v-R2}, we can calculate from \eqref{eqn:v}, \eqref{eqn:Mmotiv}, and the conservation of mass:
\begin{equation}
 \nabla_M \cdot \mathbf{v} =  \rho - C m.
\label{eqn:div-v-M}
\end{equation}

On a Riemannian manifold $(M,g)$, analogous to \eqref{eqn:div-expand}, we have
\begin{equation}
\label{eqn:div-expand-M}
\nabla_M \cdot (\rho \mathbf{v}) = \nabla_{\mathbf{v}} \rho + \rho\,  \nabla_M \cdot \mathbf{v},
\end{equation}
where $\nabla_{\mathbf{v}} \rho$ denotes the covariant derivative of $\rho$ along the tangent vector field $\mathbf{v}$. Then, one can proceed as in the Euclidean case (see \eqref{eqn:pde-expand} and \eqref{eqn:evol-charact}), and get from \eqref{eqn:pde}, \eqref{eqn:div-v-M} and \eqref{eqn:div-expand-M} an ODE for the evolution of $\rho(\Phi_t(\bfX),t)$:
\begin{equation}
\label{eqn:evol-charact-M}
\frac{D}{Dt} \rho =-\rho (\rho - C m).
\end{equation}
Here, the left-hand-side denotes the material derivative of the density along the flow $\Phi_t$ defined in \eqref{eqn:charact}:
\[
\frac{D}{Dt} \rho = \rho_t + \nabla_{\mathbf{v}} \rho.
\]

By \eqref{eqn:evol-charact-M}, densities along particle paths approach the constant value $Cm$. Therefore, as in the Euclidean case, equilibrium states have constant densities on their supports.  The primary purpose of the current paper is to apply the considerations above for two specific geometries: the two-dimensional sphere in $\R^3$ and the two-dimensional hyperbolic plane. In particular, we find interaction potentials that satisfy \eqref{eqn:Mmotiv}  and then investigate the dynamics and equilibria of the aggregation model \eqref{eqn:model} in these setups.


\section{Aggregation model on the sphere}
\label{sect:sphere}

In this section we set up the aggregation model \eqref{eqn:model} on the 2-dimensional sphere in $\R^3$. We then construct a certain interaction potential and investigate analytically and numerically the long time behaviour of the solutions.

\subsection{Model setup}
\label{subsect:setup-sphere}

Let $\bfx = x\ei + y\ej + z\ek$ denote the position of a particle in $\R^3$, where $\{\ei,\ej,\ek\}$ is the standard orthonormal Cartesian basis. 

Consider the unit sphere $S$ in $\R^3$, parametrized by spherical coordinates $(\h,\phi)$:
\begin{equation}
\label{eqn:sphere-param}
x =  \sin{\h} \cos{\phi}, \qquad y =  \sin{\h} \sin{\phi}, \qquad z = \cos{\h},
 \end{equation}
where $\h \in [0,\pi]$ is the angle from the positive $z$-axis, and $\phi \in [0, 2 \pi)$ denotes the polar angle in the $xy$-plane.

The tangent vectors at $(\h,\phi)$ on $S$ are
\begin{align*}
\bfx_\h &= \cos{\h}\cos{\phi} \, \ei + \cos{\h}\sin{\phi} \, \ej - \sin{\h}\, \ek, \\
\bfx_\phi &= -\sin{\h} \sin{\phi}\, \ei + \sin{\h} \cos{\phi}\, \ej,
\end{align*}
and  the first fundamental form (the metric) is given by:
\begin{equation}
\label{eqn:sphere-metric}
g_{11} = 1, \qquad g_{12}=0, \qquad g_{22} = \sin^2{\h}.
 \end{equation}
The metric matrix has determinant $|g| = \sin^2{\h}$ and its inverse has entries:
\begin{equation}
\label{eqn:sphere-imetric}
g^{11} = 1, \qquad g^{12}=0, \qquad g^{22} = \frac{1}{\sin^2{\h}}.
 \end{equation}



\paragraph{Gradient and Laplace-Beltrami operator on sphere.}
For a scalar function $f$ on the sphere, its surface gradient is given by (see \eqref{eqn:nabla-M}):
\begin{equation}
\label{eqn:nabla-sphere}
\nabla_S f = \frac{\p f }{\p \h}\, \bfx_\h+ \frac{1}{\sin^2 \h} \frac{\p f }{\p \phi}\, \bfx_\phi.
\end{equation}
Also, by \eqref{eqn:LB-M}, we have the expression of the Laplace-Beltrami operator in spherical coordinates:
\begin{equation}
\label{eqn:LB-sphere}
\Delta_S f = \frac{1}{\sin{\h}} \frac{\p}{\p \h}\left(\sin{\h}   \frac{\p f}{\p \h}   \right) +  \frac{1}{\sin^2{\h}} \frac{\p^2 f}{\p \phi^2}.
\end{equation}

\paragraph{Choice of interaction potential.}
We look for an interaction potential on the sphere that satisfies (see \eqref{eqn:Mmotiv}):
\begin{equation}
\label{eqn:Smotiv}
\Delta_S K(\bfx,\bfx') = -\delta_{\bfx'} + C,
\end{equation}
where $\Delta_S$ is the Laplace-Beltrami operator given by \eqref{eqn:LB-sphere}, and $C>0$ a constant. We remark from the start that since the sphere is a closed manifold, the constant $C$ cannot be arbitrary; by the solvability condition for \eqref{eqn:Smotiv}, $C = \frac{1}{4 \pi}$. 

Based on the observation above, we choose $K$ to be the Green's function in the generalized sense of $\Delta_S$ \cite{Kimura1999}, i.e., we set
\begin{equation}
\label{eqn:K-sphere}
K(\bfx,\bfx') = G_S(\bfx;\bfx'),
\end{equation}
where
\begin{equation}
\label{eqn:G-sphere}
G_S(\bfx; \bfx') = -\frac{1}{2 \pi} \log{\sin \left( \frac{d(\bfx,\bfx')}{2}\right)}.
\end{equation}
Here, $d(\bfx, \bfx')$ denotes the spherical distance between the points $\bfx$ and $\bfx'$. It is a simple exercise to check indeed that $G_S$ satisfies:
\begin{equation}
\label{eqn:SLaplacian}
\Delta_S \, G_S(\bfx; \bfx') = -\delta_{\bfx'} + \frac{1}{4 \pi}.
\end{equation}


In local coordinates $(\h,\phi)$, $(\h',\phi')$ (corresponding to $\bfx$ and $\bfx'$, respectively), the spherical distance is given by the law of cosines on sphere:
\begin{equation}
\cos{d(\bfx, \bfx')} = \cos{\h}\cos{\h'} + \sin{\h}\sin{\h'}\cos{(\phi - \phi')}.
\end{equation}
Also, using local coordinates in \eqref{eqn:conv} we have:
\begin{equation}
\label{eqn:conv-S}
K * \rho (\h,\phi) = \int_{0}^{2\pi} \int_{0}^{\pi} K(\h, \phi, \h', \phi') \rho(\h',\phi')  \sin{\h'} d\h' d\phi',
\end{equation}
where by an abuse of notation, we used $K(\h, \phi, \h', \phi')$ for $K(\bfx,\bfx')$.

The interaction potential $K$ is {\em purely} repulsive. Indeed, since $K(\bfx,\bfx')$ only depends on the geodesic distance between the points, it is enough to take $\bfx'$ as the North pole of the sphere ($(0,0)$ in local coordinates). Then, $d(\bfx,\bfx') = \h$ and by \eqref{eqn:v}, \eqref{eqn:nabla-sphere} and the choice of $K$ in \eqref{eqn:K-sphere}-\eqref{eqn:G-sphere}, upon interacting with the North pole, the point $\bfx$ moves in the direction 
\begin{equation}
\label{eqn:repel-sphere}
- \nabla_S K(\h,\phi,0,0) = \underbrace{\frac{1}{4 \pi} \operatorname{cotan}\left( \frac{\h}{2} \right)}_{>0 \text{ (repulsion)} } \bfx_\h.
\end{equation}
Hence,  all points on $S$ sense a repelling force (except the South pole for which the expression above vanishes) -- see Figure \ref{fig:att-rep}(a) for an illustration.

\begin{figure}[htb]
 \begin{center}
 \begin{tabular}{cc}
 \includegraphics[width=0.5\textwidth]{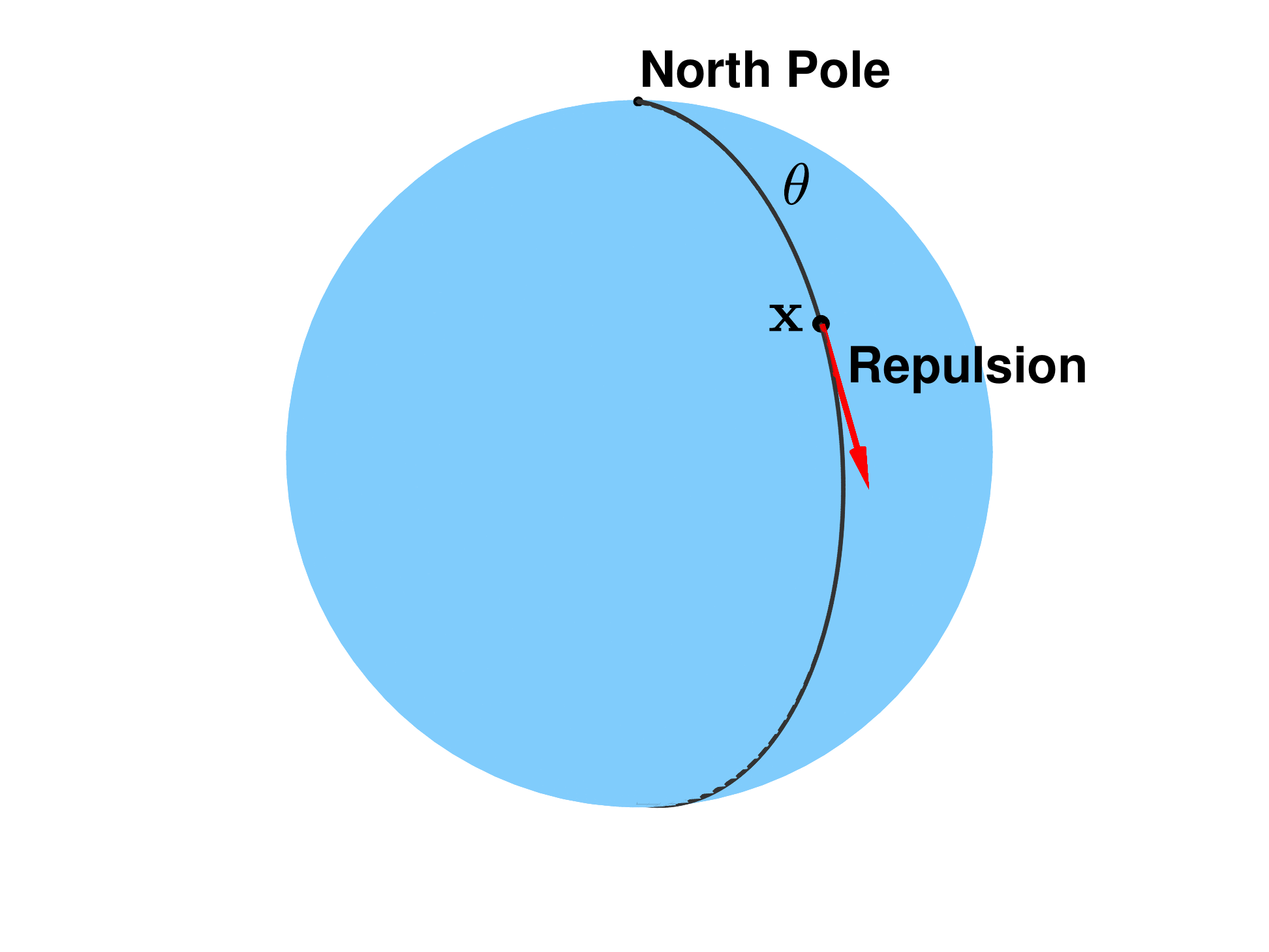} &
 \includegraphics[width=0.5\textwidth]{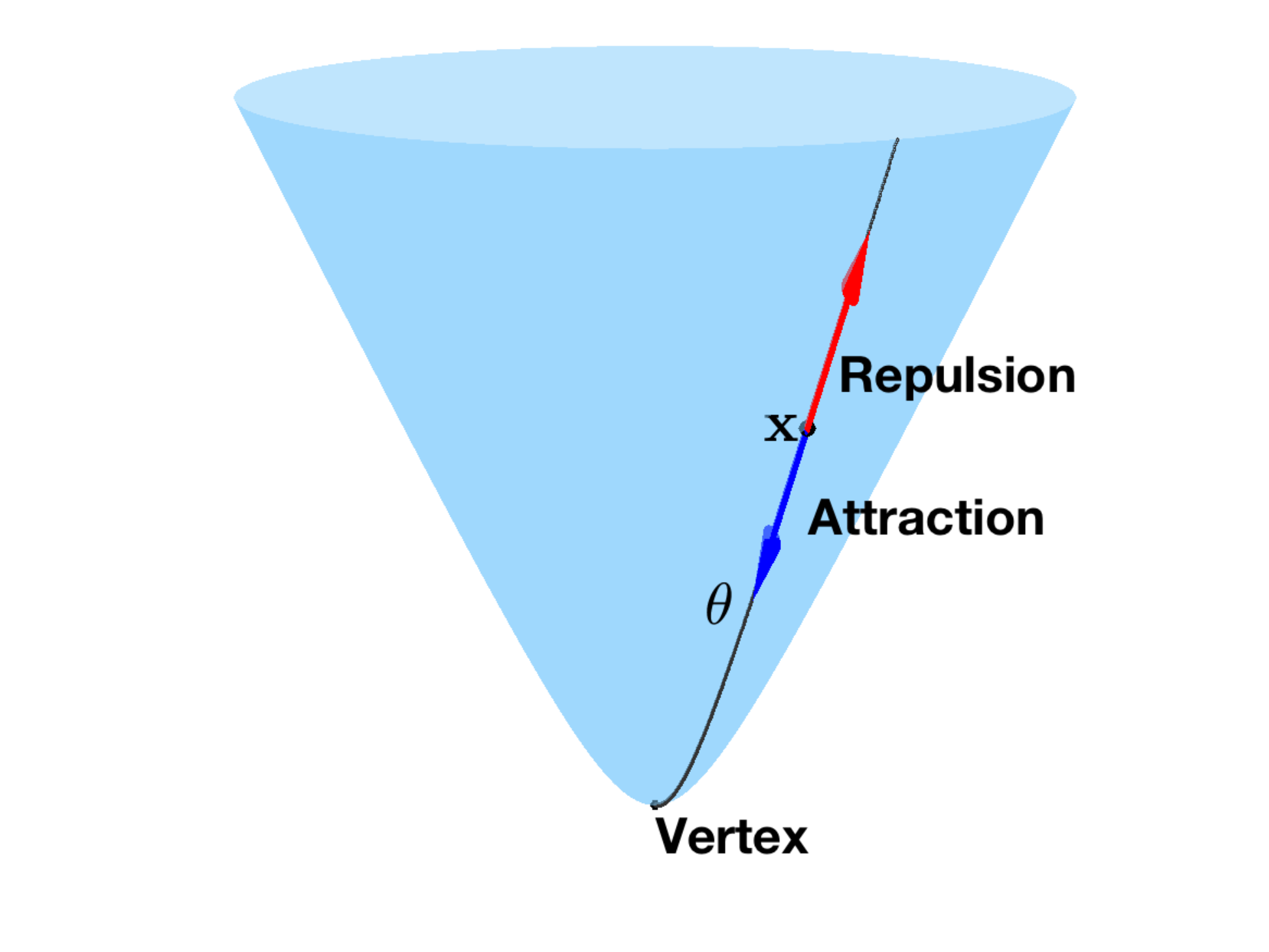} \\
 (a) & (b)
 \end{tabular}
\caption{(a) The interaction potential \eqref{eqn:K-sphere} on $S$ is purely repulsive. The figure indicates how a generic point senses a repelling force from the North pole -- see \eqref{eqn:repel-sphere}. (b) The interaction potential \eqref{eqn:K-hyper} on $H$ is attractive-repulsive, as the two terms in the right-hand-side of \eqref{eqn:K-hyper} have competing effects. Shown in the figure is a generic point interacting with the vertex of the hyperboloid -- see \eqref{eqn:repel-hyper}.}
\label{fig:att-rep}
\end{center}
\end{figure}


\subsection{Asymptotic convergence to equilibrium} 
\label{subsect:asympt-sphere}
For simplicity, set the total mass $m=1$. From \eqref{eqn:v}, \eqref{eqn:K-sphere}, \eqref{eqn:SLaplacian} and the conservation of mass, we find (see also \eqref{eqn:div-v-M}):
\begin{equation}
\label{eqn:div-v-sphere}
\nabla_S \cdot \mathbf{v} 
= \rho - \frac{1}{4 \pi}. 
\end{equation}

Hence, along the flow $\Phi_t$ generated by the vector field $\mathbf{v}$ on $S$, the density $\rho(\Phi_t(\bfX),t)$ satisfies (see \eqref{eqn:evol-charact-M}):
\begin{equation}
\label{eqn:rho-flow-sphere}
\frac{D}{Dt} \rho =-\rho \left(\rho - \frac{1}{4\pi}\right).
\end{equation}

As noted above, along particle paths, densities $\rho(\Phi_t(\bfX),t)$ approach a constant value (here, the constant density at equilibrium is $\frac{1}{4 \pi}$). The ODE \eqref{eqn:rho-flow-sphere} can be solved exactly in fact, with solution:
\begin{equation}
\rho(\Phi_t(\bfX),t)= \frac{1}{ 4 \pi +\left(  \frac{1}{\rho_{0}(\bfX)}- 4 \pi \right) e^{-\frac{t}{4 \pi}}},
\label{eqn:rho-c}
\end{equation}
where $\rho_{0}$ is the initial density.

Also, from \eqref{eqn:continuity-alt} and \eqref{eqn:rho-c}, one can find an exact explicit expression for the Jacobian of the flow map:
\begin{equation}
\label{eqn:J}
J(\bfX,t) = 4 \pi \rho_{0}(\bfX) + \left(  1 -
4 \pi \rho_{0}(\bfX) \right)  e^{-\frac{t}{4\pi}},
\end{equation}
where for notational convenience we have denoted  
\begin{equation}
\label{eqn:J-particle}
J(\bfX,t) := J(\Phi_t)(\bfX).
\end{equation}
Note that
\[
J(\bfX,t) >0 \qquad\text{ for all } t,
\]
guaranteeing that the particle map is invertible for as long as it exists.

\begin{thm}[Global attractor for sphere]
\label{thm:sphere}
Consider model \eqref{eqn:model} set up on the sphere $S$, with the interaction potential $K$ given by \eqref{eqn:K-sphere}-\eqref{eqn:G-sphere}. Assume that the model has a global in time $C^1$ solution $\rho(\bfx,t)$ and that the flow map $\Phi_t:S \to S$ is $C^1$ and invertible for all $t>0$. Then, solutions $\rho(\bfx,t)$ approach asymptotically, as $t\to \infty$, a constant equilibrium density supported over the entire sphere.
\end{thm}
\begin{proof}
The proof is essentially provided by the considerations above. From \eqref{eqn:rho-flow-sphere} and \eqref{eqn:rho-c}, along particle paths that originate from the support of $\rho_0$ (i.e., $\rho_0(\bfX) \neq 0$), one has
\[
\lim_{t \to \infty} \rho(\Phi_t(\bfX),t) = \frac{1}{4 \pi}.
\]
Consequently,  equilibria must have constant density $\frac{1}{4 \pi}$ on their support. And since the total mass $m=1$ is conserved through the time evolution, then necessarily the support of the equilibrium state must be the entire sphere (up to a zero measure set). This shows that
\begin{equation}
\label{eqn:sstate-sphere}
\bar{\rho}_S(\bfx) = \frac{1}{4 \pi} \qquad \qquad \text{ for all } \bfx \in S,
\end{equation}
is a global attractor.
\end{proof}


\paragraph{Symmetric initial density.} When the initial density is symmetric with respect to rotations about a North-South axis, we can find explicit expressions for the evolution of the particle paths, and hence, show directly the result in Theorem \ref{thm:sphere}. While the calculation is interesting in itself, it is key to investigating attractors on the hyperbolic plane (Section \ref{subsect:asympt-hyper}).

For simplicity, assume the initial density is symmetric with respect to rotations about the $z$-axis. In spherical coordinates, this amounts to considering an initial density $\rho_0$ that depends only on $\h$, but not on $\phi$; by an abuse of notation we will denote $\rho_0(\bfX)$ by $\rho_0(\h)$. Note that a symmetric initial density results in a symmetric solution for all times. Indeed, for the initial density $\rho_0(\h)$, evaluate the velocity at a generic point $(\h,\phi)$ using \eqref{eqn:v} and \eqref{eqn:conv-S}. One finds that the velocity only depends on $\h$ and also, that it points in the direction $\bfx_\h$, i.e., the point moves along the meridian $\phi=\text{const}$. The same argument applies to a symmetric density at an arbitrary time instance, hence the symmetry is preserved through time evolution.

Note that in spherical coordinates (see \eqref{eqn:vol-form}), one has 
\begin{equation}
\label{eqn:mu-sphere}
\mu = \sin \h \, d\h \wedge d \phi.
\end{equation}
By symmetry,  the flow map in local coordinates takes a point $\bfX \in S$ of coordinates $(\h,\phi$) into $\Phi_t(\bfX) \in S$ of coordinates $(\lambda(\h,t),\phi)$, that is, the coordinate $\phi$ remains constant and the coordinate $\h$ maps into $\lambda(\h,t)$, for some function $\lambda$.

Then, by \eqref{eqn:J-flow} and \eqref{eqn:mu-sphere} (see also notation \eqref{eqn:J-particle}), we find in local coordinates that: 
\begin{equation}
\label{eqn:J-lambda}
 \sin{(\lambda(\h,t))} \lambda_\h(\h,t) = J(\bfX,t) \, {\sin \h}.
\end{equation}
Note that, as expected, the Jacobian depends only on $\h$ and not on $\phi$, as the latter coordinate remains constant along the flow. 
Now, we get from \eqref{eqn:J} and \eqref{eqn:J-lambda} the following differential equation for $\lambda(\h,t)$:
\[
\sin{(\lambda(\h,t))} \lambda_\h(\h,t) =  \left( 4 \pi \rho_{0}(\h) + \left(  1 -
4 \pi \rho_{0}(\h) \right)  e^{-\frac{t}{4\pi}}\right) \sin \h.
\]
The ODE above can be solved by elementary methods. We find:
\begin{equation}
\label{eqn:flow-exact-sphere}
\cos (\lambda(\h,t)) = 1 - 4 \pi \int_0^\h \rho_{0}(\xi) \sin \xi \, d\xi -  e^{-\frac{t}{4\pi}} \int_0^\h \left(  1 - 4 \pi \rho_{0}(\xi) \right) \sin \xi \, d\xi,
\end{equation}
where we used the fact that by symmetry, the North pole ($\h=0$) remains fixed (i.e., $\lambda(0,t)=0$ for all $t$). Equation \eqref{eqn:flow-exact-sphere} provides an exact explicit expression for the flow map $(\h,\phi) \to (\lambda(\h,t),\phi)$ in spherical coordinates. 

To find the asymptotic behaviour of the particle trajectories, one can send $t \to \infty$ in equation \eqref{eqn:flow-exact-sphere}. We find:
\begin{equation}
\label{eqn:cosLambda}
\cos {\Lambda_\h} = 1- 4 \pi \int_0^\h \rho_{0}(\xi) \sin \xi \, d\xi.
\end{equation}
where 
\[
\Lambda_\h = \lim_{t \to \infty} \lambda(\h,t).
\]
Note that the (conserved) unit mass can be written in spherical coordinates as:
\begin{equation}
\label{eqn:m1spherical}
1 = \int_0^{2 \pi} \int_0^\pi \rho_{0}(\h) \sin \h \, d\h d \phi.
\end{equation}

Now consider a symmetric domain $ \theta_0 \leq \theta \leq \theta_1$ that contains the support of $\rho_0$. Then, for any particle trajectory that originates from $\theta<\theta_0$ (outside the support of $\rho_0$), we have by \eqref{eqn:cosLambda} that $\Lambda_\h = 0$. In other words, all trajectories starting from $\theta<\theta_0$ approach the North pole as $t \to \infty$. Similarly, for any particle trajectory that originates from $\theta>\theta_1$ (also outside the initial support), one finds by \eqref{eqn:cosLambda}  and \eqref{eqn:m1spherical} that $\Lambda_\h = \pi$, so the trajectory approaches the South pole as $t \to \infty$. On the other hand, by \eqref{eqn:cosLambda}, $\Lambda_\theta$ is monotonic and continuous in $\theta$, so trajectories starting from inside the initial support will spread over the entire sphere, as expected by the result in Theorem \ref{thm:sphere} -- see Figure \ref{fig:sphere-sym} for a numerical illustration using a particle method.

\begin{rmk}
\label{rmk:symmetry-sphere} The considerations above hold for initial densities that are symmetric with respect to rotations about any North-South axis of the sphere. The same exact expressions of the particle trajectories hold, upon a rotation of the coordinate axes.
\end{rmk}

\begin{figure}[htb]
 \begin{center}
 \begin{tabular}{cc}
 \includegraphics[width=0.5\textwidth]{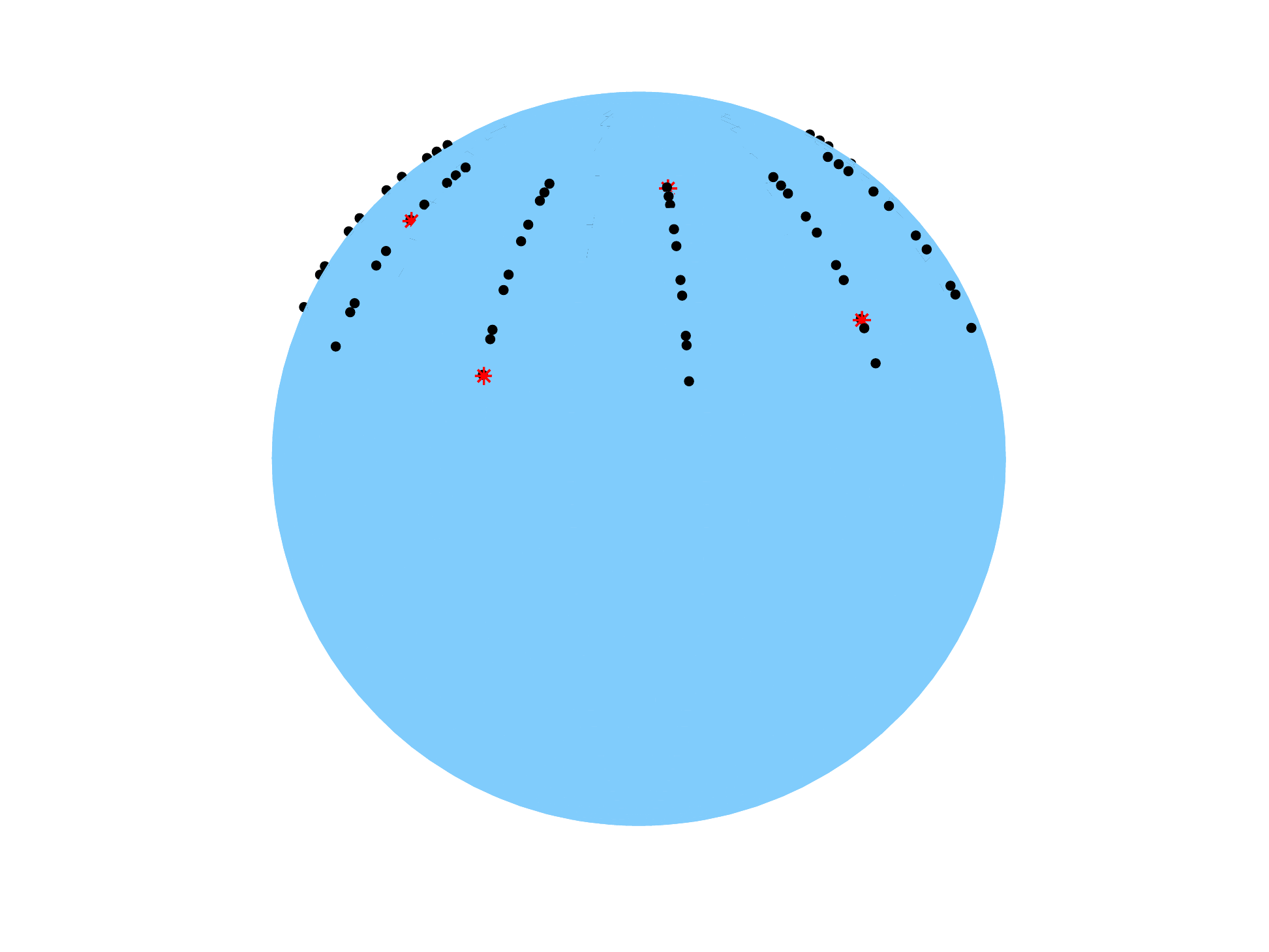}  &
 \includegraphics[width=0.5\textwidth]{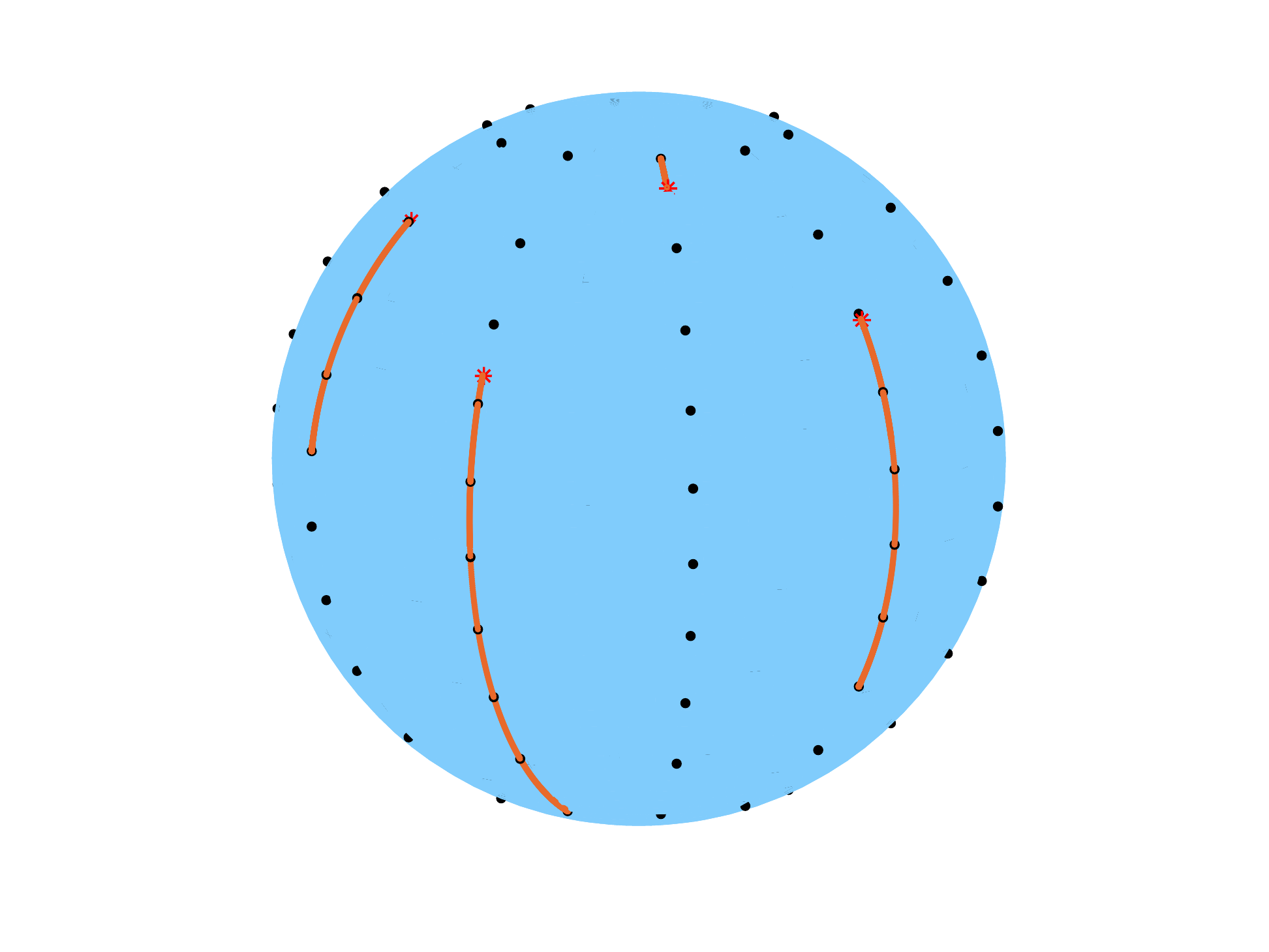} \\
 (a) & (b)
 \end{tabular}
\caption{ Numerical simulation with $N=100$ particles for model \eqref{eqn:model} on $S$ with interaction potential given by \eqref{eqn:K-sphere}-\eqref{eqn:G-sphere}. (a) Symmetric initial configuration on $S$, with $\theta$ coordinates generated randomly in the interval $\left( \frac{\pi}{8}, \frac{3 \pi}{8} \right) $.  (b) The configuration remains symmetric for all times and evolves into a uniform particle distribution supported over the entire sphere -- see \eqref{eqn:sstate-sphere}. The solid lines represent the trajectories of the particles indicated by stars in figure (a).}
\label{fig:sphere-sym}
\end{center}
\end{figure}


\subsection{Numerical results}
\label{subsubsect:numerics-sphere}

\paragraph{Numerical method.} We will use the discrete particle system associated to the macroscopic model \eqref{eqn:model}. 
Set the total mass $m=1$ and consider $N$ particles of equal mass (so each particle has mass $\frac{1}{N}$). Let $\bfx_i(t)$ represent the location in $\R^3$ of the $i$-th particle.  Equation \eqref{eqn:v} can be written in discrete form to express the velocity of particle $\bfx_i$ in terms of the locations $\bfx_j$ ($j \neq i$) of the other particles. Hence, one arrives at the discrete particle system:
\begin{equation}
\label{eqn:discrete}
\frac{d \bfx_i}{dt} = -\frac{1}{N} \sum_{\substack{j = 1 \\ j \ne i}}^N \nabla_{S,i} K(\bfx_i,\bfx_j), \qquad i =1,\dots,N.
\end{equation}

In Euclidean settings, the rigorous mean-field limit of the particle system \eqref{eqn:discrete} was established in \cite{CaChHa2014}. Specifically, it was shown that the empirical distribution associated to \eqref{eqn:discrete}, i.e., $\mu_N(t) = \frac{1}{N} \sum_{i=1}^N \delta_{\bfx_i(t)}$, converges weak-$*$ as measures to a weak solution $\rho(t)$ of the macroscopic model \eqref{eqn:model}. The result holds for a general class of potentials, including repulsive-attractive potentials that have a (strictly better than Newtonian) singularity at origin.

We write both sides of \eqref{eqn:discrete} in the local spherical basis at $\bfx_i$. By chain rule, the left-hand-side can be expanded as:
\begin{equation*}
\frac{d\bfx_i}{dt} = \frac{\p \bfx_i}{\p \h_i}\frac{d\h_i}{dt} + \frac{\p \bfx_i}{\p \phi_i}\frac{d\phi_i}{dt},
\end{equation*}
while for the right-hand-side we use \eqref{eqn:nabla-sphere} to get:
\begin{equation*}
\nabla_{S,i} K (\bfx_i,\bfx_j) = \frac{\p K}{\p \hi}  \frac{\p \bfx_i}{\p \h_i} + \frac{1}{\sin^2{\hi}} \frac{\p K}{\p \phi_i}  \frac{\p \bfx_i}{\p \phi_i} .
\end{equation*}
By matching the coefficients on each side, we obtain the following ODE system for the spherical coordinates:
\begin{equation}
\label{eqn:ode-sphere}
\begin{aligned}
\frac{d \hi}{dt} &= -\frac{1}{N}\sum_{\substack{j=1 \\ j \neq i}}^N  \frac{\p K}{\p \hi} (\h_i,\phi_i,\h_j,\phi_j), \qquad i=1,\dots,N\\
\frac{d \phi_i}{dt} &=-\frac{1}{N}\sum_{\substack{j=1 \\ j \neq i}}^N { \frac{1}{\sin^2{\hi}} \frac{\p K}{\p \phi_i}} (\h_i,\phi_i,\h_j,\phi_j), \qquad i = 1,\dots,N.
\end{aligned}
\end{equation}

The results presented below are obtained by solving numerically the particle system \eqref{eqn:ode-sphere} with the classical (4th order) Runge-Kutta method. We note that the particle method is very suitable here, as it complements the Lagrangian approach from Section \ref{subsect:asympt-sphere}. Simulating the discrete system is in fact the main numerical tool used for the model in $\R^n$.


\paragraph{Numerical simulations on sphere.} We use the particle method to confirm the theoretical findings from Section \ref{subsect:asympt-sphere}. First, we consider an initial  configuration which is symmetric about the $z$-axis (Figure  \ref{fig:sphere-sym}(a)). To generate this initial density, we first placed particles randomly along a given meridian ($\phi=\text{const.}$, $ \frac{\pi}{8} < \theta < \frac{3 \pi}{8}$), and then rotated the configuration about the $z$-axis. We find indeed that particles evolve along meridians (see the trajectories of the particles indicated by stars) into a uniform particle distribution supported over the entire sphere. 

The equilibrium \eqref{eqn:sstate-sphere} is a global attractor (Theorem \ref{thm:sphere}). Figure \ref{fig:sphere} shows a numerical validation of this result: a random initial particle distribution (Figure \ref{fig:sphere}(a)) evolves into the expected equilibrium state (Figure \ref{fig:sphere}(b)). The $\theta$ and $\phi$ coordinates of the initial configuration were drawn randomly from the intervals $ \left( \frac{\pi}{8}, \frac{3 \pi}{8} \right)$ and $ \left(0, \frac{\pi}{2} \right)$, respectively. The solid lines in Figure \ref{fig:sphere}(b) represent several individual trajectories,  corresponding to the particles indicated by stars.

\begin{figure}[htb]
 \begin{center}
 \begin{tabular}{cc}
 \includegraphics[width=0.5\textwidth]{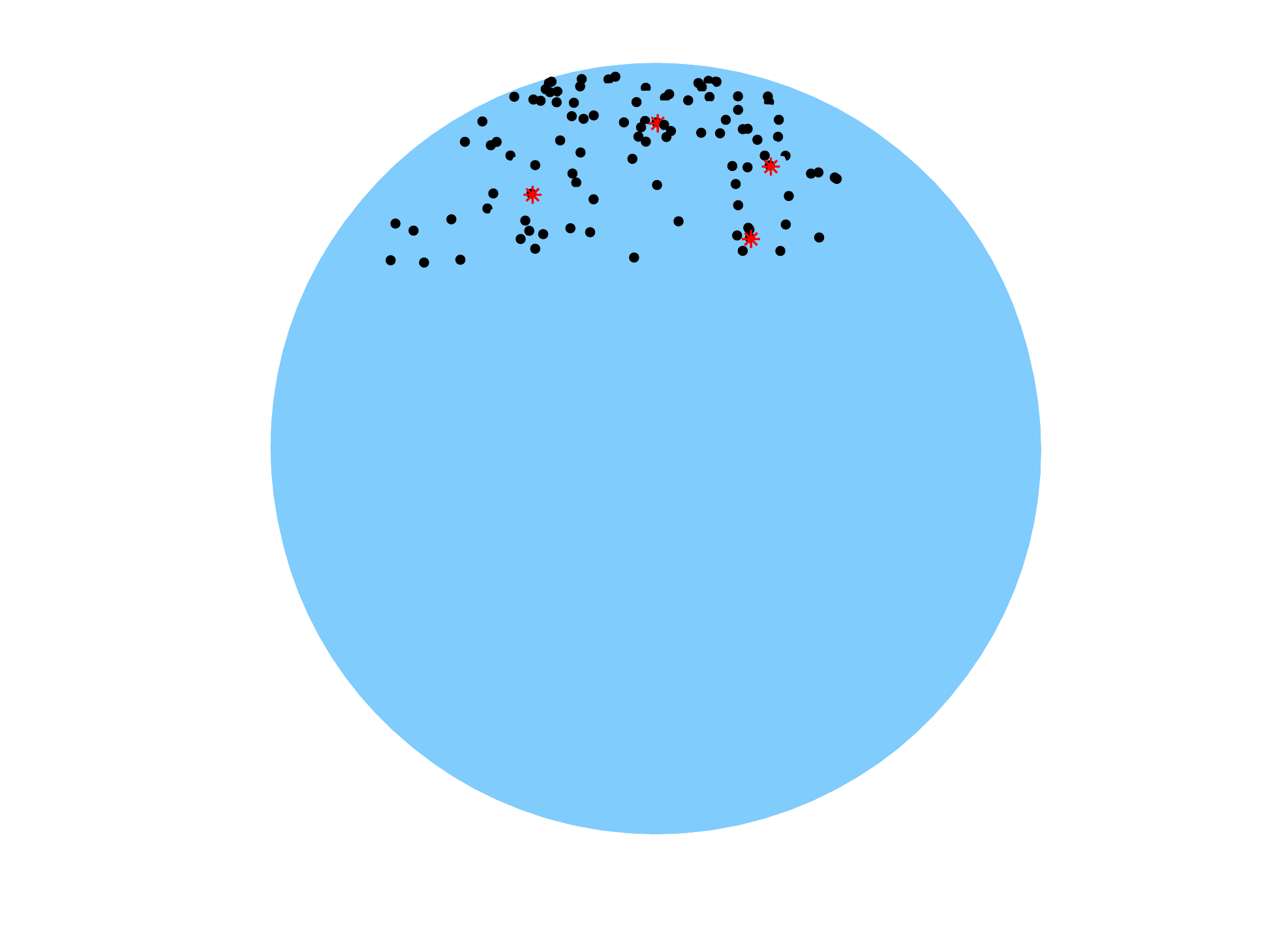} &
 \includegraphics[width=0.5\textwidth]{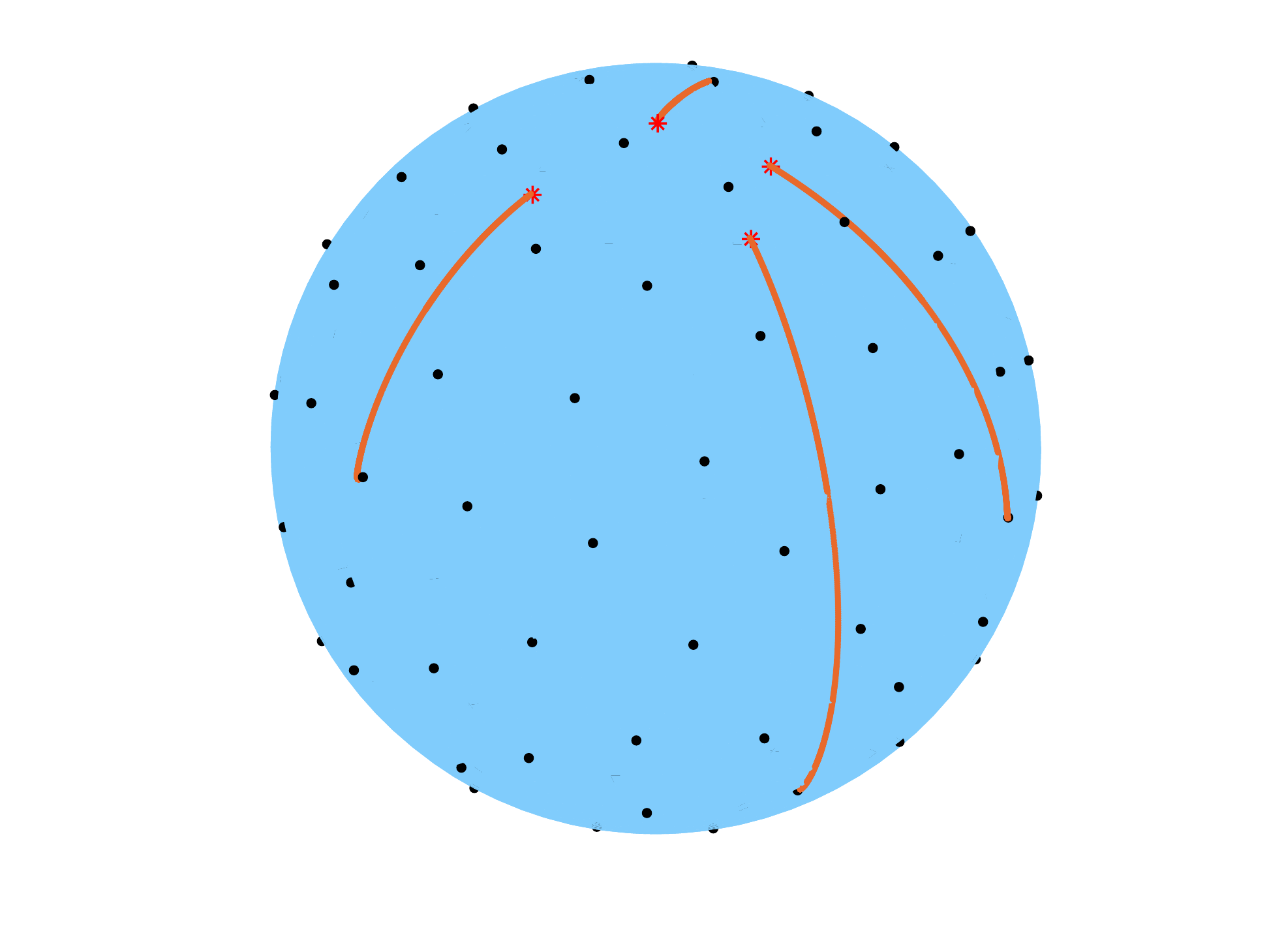} \\
 (a) & (b)
 \end{tabular}
\caption{Numerical simulation with $N=100$ particles for model \eqref{eqn:model} on $S$ with interaction potential given by \eqref{eqn:K-sphere}-\eqref{eqn:G-sphere}. (a) Random initial configuration on $S$, with coordinates $\theta$ and $\phi$ generated randomly in $\left( \frac{\pi}{8}, \frac{3 \pi}{8} \right)$ and $ \left(0, \frac{\pi}{2} \right)$, respectively.  (b) The configuration in (a) evolves into a uniform particle distribution supported over the entire sphere -- see Theorem \ref{thm:sphere}. The solid lines represent the trajectories of the particles indicated by stars in figure (a).}
\label{fig:sphere}
\end{center}
\end{figure}


\section{The aggregation model on the hyperbolic space}
\label{sect:hyperboloid}

In this section we set up the aggregation model on the 2-dimensional hyperbolic space and investigate the dynamics of its solutions for interaction potentials that lead to equilibria of constant densities. 

\subsection{Model setup}
\label{subsect:setup-hyper}

We use the hyperboloid model of the two dimensional hyperbolic space \cite{CaFlKePa1997}. Specifically, we consider the upper sheet of the two-sheeted hyperboloid:
\[
H = \{ (x,y,z) \in \R^3 \mid x^2 + y^2 - z^2 = -1 \text{ and } z>0\},
\]
embedded in $\R^3$ endowed with the Minkowski inner product 
\[
\langle \bfx, \bfx'\rangle = x x' + y y' - z z'.
\]
Here, $\bfx = x \ei + y\ej + z\ek$ and $\bfx' = x' \ei + y' \ej + z' \ek$.

The hyperboloid $H$ (in the pseudo-Euclidean space $\R^3$) can be parametrized as:
\begin{equation}
\label{eqn:hyper-param}
x =  \sinh{\h} \cos{\phi}, \qquad y =  \sinh{\h} \sin{\phi}, \qquad z = \cosh{\h}.
 \end{equation}
where $\h \in [0,\infty)$ can be identified with the hyperbolic distance from the vertex of the hyperboloid (i.e., point (0,0,1)) and $\phi \in [0, 2 \pi)$ denotes the polar angle in the $xy$-plane.

The tangent vectors at $(\h,\phi)$ on $H$ are
\begin{align*}
\bfx_\h &= \cosh{\h}\cos{\phi} \, \ei + \cosh{\h}\sin{\phi} \, \ej + \sinh{\h}\, \ek, \\
\bfx_\phi &= -\sinh{\h} \sin{\phi}\, \ei + \sinh{\h} \cos{\phi}\, \ej,
\end{align*}
and the metric coefficients are given by:
\begin{equation}
\label{eqn:hyper-metric}
g_{11} = 1, \qquad g_{12}=0, \qquad g_{22} = \sinh^2{\h}.
 \end{equation}
The determinant of the metric is $|g| = \sinh^2{\h}$ and its inverse given by:
\begin{equation}
\label{eqn:hyper-imetric}
g^{11} = 1, \qquad g^{12}=0, \qquad g^{22} = \frac{1}{\sinh^2{\h}}.
 \end{equation}

The hyperbolic distance $d(\bfx,\bfx')$ between two points $\bfx$ and $\bfx'$ on $H$, of local coordinates $(\h,\phi)$ and $(\h',\phi')$, respectively,  can be found from the hyperbolic law of cosines:
\begin{equation}
\label{eqn:dist-H}
\cosh{d(\bfx,\bfx')} =  \cosh{\h}\cosh{\h'} - \sinh{\h}\sinh{\h'}\cos{(\phi - \phi')}.
\end{equation}



\paragraph{Gradient and Laplace-Beltrami operators on $\mathbf{H}$.} From \eqref{eqn:nabla-M} and \eqref{eqn:LB-M} we find:
\begin{equation}
\label{eqn:nabla-H}
\nabla_H f = \frac{\p f}{\p \h}\,  \bfx_\h + \frac{1}{\sinh^2{\h}} \frac{\p f}{\p \phi}\,  \bfx_\phi,
\end{equation}
and
\begin{equation}
\label{eqn:LB-H}
\Delta_H f = \frac{1}{\sinh{\h}} \frac{\p}{\p \h}\left(\sinh{\h}   \frac{\p f}{\p \h}   \right) +  \frac{1}{\sinh^2{\h}} \frac{\p^2 f}{\p \phi^2}. 
\end{equation}


\paragraph{Choice of interaction potential.} The Green's function of the negative Laplacian on $H$ is given by (see \cite{Kimura1999}):
\begin{equation}
\label{eqn:Green-hyper}
G_H(\bfx;\bfx') = - \frac{1}{2\pi} \log  { \tanh{ \left( \frac{d(\bfx,\bfx')}{2} \right)}} .
\end{equation}
Indeed, using \eqref{eqn:LB-H} and \eqref{eqn:dist-H}, one can check that $-\Delta_H G(\bfx;\bfx') = \delta_{\bfx'}$ holds in the sense of distributions.

Motivated by the considerations made in Section \ref{sect:prelim}, we look now for a function with positive constant Laplacian.  By elementary methods (using \eqref{eqn:LB-H}), one finds that
\begin{equation}
\label{eqn:A}
A(\h) =  2 \log{ \cosh{\left( \frac{\h}{2} \right)}}
\end{equation}
satisfies
\[
\Delta_H A(\h) = 1.
\]
Working in more generality, and using the hyperbolic distance between points, it also holds that:
\[
\Delta_H A(d(\bfx,\bfx')) = 1.
\]
Consequently, we choose the interaction potential as 
\begin{equation}
\label{eqn:K-hyper}
K(\bfx,\bfx') = G_H(\bfx;\bfx') + A(d(\bfx,\bfx')),
\end{equation}
which by the considerations above satisfies
\begin{equation}
\label{eqn:KLaplacian}
\Delta_H K(\bfx,\bfx') =  -\delta_{\bfx'} + 1.
\end{equation}

\begin{rmk}
\label{remark:gamma}
We note that one can multiply the function $A$ by any positive constant and reach a model with similar properties. Indeed, the interaction potential
\begin{equation*}
K(\bfx,\bfx') = G_H(\bfx;\bfx') + C \cdot A(d(\bfx,\bfx')),
\end{equation*}
satisfies
\begin{equation*}
\Delta_H K(\bfx,\bfx') = -\delta_{\bfx'} + C.
\end{equation*}
Consequently, any $C>0$ serves the purpose of the present study, which is to have a model that evolves into an equilibrium of constant density (the constant density in this case would be $Cm$, see \eqref{eqn:evol-charact-M}). 
\end{rmk}

Using local coordinates, \eqref{eqn:conv} becomes
\[
K * \rho (\h,\phi) = \int_{0}^{2\pi} \int_{0}^{\infty} K(\h, \phi, \h', \phi') \rho(\h',\phi') \sinh{\h'} d\h' d\phi',
\]
where we abused notation and wrote $K(\h, \phi, \h', \phi')$ for $K(\bfx,\bfx')$.

To determine the nature of the potential \eqref{eqn:K-hyper} we proceed as in the sphere case. With no loss of generality, we take $\bfx'$ to be the vertex of the hyperboloid (of local coordinates $(0,0)$) and $\bfx = (\h,\phi)$ a generic point on $H$, so that $d(\bfx,\bfx') = \h$. From \eqref{eqn:v}, \eqref{eqn:nabla-H} and \eqref{eqn:K-hyper} (see also \eqref{eqn:Green-hyper} and \eqref{eqn:A}), by interacting with the vertex, the point $\bfx$ moves in the direction 
\begin{equation}
\label{eqn:repel-hyper}
- \nabla_H K(\h,\phi,0,0) = \underbrace{\frac{1}{2 \pi} \frac{1}{\sinh \h}}_{>0 \text{ (repulsion)}} \bfx_\h \underbrace{ -\tanh \frac{\h}{2}}_{<0 \text{ (attraction)}}\bfx_\h.
\end{equation}
Specifically, the first term in the right-hand-side of \eqref{eqn:repel-sphere}, resulting from the Green's function component of $K$, is repulsive, while the second term, corresponding to term $A(\h)$ of $K$, is attractive -- see Figure \ref{fig:att-rep}(b) for an illustration. Repulsion and attraction balance each other at $\theta = \operatorname{cosh}^{-1}\left( 1+\frac{1}{2 \pi} \right)$, while the net effect of the two interactions is repulsive or attractive at distances $\theta$ smaller or larger than this value, respectively (short range repulsion and long range attraction).

\begin{rmk}
\label{remark:ar-similarity}
The attractive-repulsive character of the interactive potential \eqref{eqn:K-hyper} mimics very well that of the interaction potential \eqref{eqn:K-plane} from the Euclidean case. Namely, it contains a repulsive component due to the Green's function of the negative Laplacian (on $H$), and an attractive component that is smooth and unbounded (growing at infinity). Remarkably, as shown next, this similar structure results in similar long-time behaviours of solutions of the two models, that is, to approach asymptotically equilibria of constant densities supported on geodesic disks.
\end{rmk}


\subsection{Asymptotic convergence to equilibrium} 
\label{subsect:asympt-hyper}

Set the (conserved) total mass $m=1$. Similar to the Euclidean and the sphere cases, from \eqref{eqn:v}, \eqref{eqn:K-hyper}, \eqref{eqn:KLaplacian} and the conservation of mass, we have (see \eqref{eqn:div-v-M}):
\begin{equation}
\label{eqn:div-v-hyper}
\nabla_H \cdot \mathbf{v} = \rho - 1. 
\end{equation}

Hence, along the flow $\Phi_t(\bfX)$, the density $\rho(\Phi_t(\bfX),t)$ evolves according to (see \eqref{eqn:evol-charact-M}):
\begin{equation}
\label{eqn:rho-flow-hyper}
\frac{D}{Dt} \rho =-\rho \left(\rho - 1\right).
\end{equation}
We conclude from \eqref{eqn:rho-flow-hyper} that densities $\rho(\Phi_t(\bfX),t)$ approach the constant value $1$ along particle trajectories, and consequently, the equilibrium is constant ($=1$) on its support.

As for the sphere,  \eqref{eqn:rho-flow-hyper} can be solved exactly, for an initial density $\rho_0$:
\begin{equation}
\rho(\Phi_t(\bfX),t)= \frac{1}{ 1 +\left(  \frac{1}{\rho_{0}(\bfX)}- 1 \right) e^{-t}}.
\label{eqn:rho-h}
\end{equation}

Also, from \eqref{eqn:continuity-alt} and \eqref{eqn:rho-h}, we find the Jacobian of the flow map:
\begin{equation}
\label{eqn:J-h}
J(\bfX,t) = \rho_{0}(\bfX) + \left(  1 - \rho_{0}(\bfX) \right)  e^{-t},
\end{equation}
where we have used the notation \eqref{eqn:J-particle}. Note again that $J(\bfX,t) >0$ for all times, so the particle map is invertible for as long as it exists.

Regarding the equilibrium configuration, by conservation of mass, we can infer immediately that the surface area of the equilibrium density is $1$. However, \eqref{eqn:rho-flow-hyper} does not provide any information about the shape of the equilibrium's support. As opposed to the sphere case, where we showed the asymptotic convergence for arbitrary initial densities (Theorem \ref{thm:sphere}), for the hyperboloid we will only prove asymptotic convergence for {\em symmetric} initial data. The result is the following theorem.

\begin{thm} [Attractor for hyperbolic plane: symmetric initial density]
\label{thm:hyper} Consider model \eqref{eqn:model} on $H$,
with the interaction potential  \eqref{eqn:K-hyper}. Take a compactly supported initial density $\rho_0$ on $H$ that is symmetric\footnote{We call a density symmetric about a point if it is constant on geodesic circles centred at that point.} about the vertex of the hyperboloid, and assume that model \eqref{eqn:model} with this initial data has a global in time $C^1$ solution $\rho(\bfx,t)$, with a flow map $\Phi_t:H \to H$ that is $C^1$ and invertible for all $t>0$.  Then, $\rho(\bfx,t)$ approaches asymptotically, as $t\to \infty$, a constant equilibrium density supported over a geodesic disk centred at the vertex.
\end{thm}
\begin{proof}
We follow very similar considerations as in Section \ref{subsect:asympt-sphere}. The flow map $\Phi_t:H \to H$ (assumed to be smooth and invertible) has a Jacobian $J(\Phi_t)$ that satisfies \eqref{eqn:J-flow} \cite[Proposition 7.1.10]{AMR1988}. Here, $\mu$ represents the canonical volume form on $H$ which in local coordinates \eqref{eqn:hyper-param} is given by:
\begin{equation}
\label{eqn:mu-hyper}
\mu = \sinh \h \, d\h \wedge d \phi.
\end{equation}

Using an argument similar to that for the sphere, one infers that an initial density that is symmetric about the vertex results in a symmetric solution for all times (as points move along meridians with velocities which only depend on their $\h$ coordinate). Hence,  the flow map in local coordinates takes a point $\bfX \in H$ of coordinates $(\h,\phi$) into $\Phi_t(\bfX) \in H$ of coordinates $(\lambda(\h,t),\phi)$, for some function $\lambda$. Following exactly the same steps as in Section \ref{subsect:asympt-sphere}, we find from \eqref{eqn:J-flow} and \eqref{eqn:mu-hyper}:
\begin{equation}
\label{eqn:J-lambda-hyper}
 \sinh{(\lambda(\h,t))} \lambda_\h(\h,t) = J(\bfX,t) \, {\sinh \h},
\end{equation}
and further, by using \eqref{eqn:J-h}, we get a differential equation for $\lambda$:
\begin{equation}
\label{eqn:J-lambda-mod}
\sinh{(\lambda(\h,t))}  \lambda_\h(\h,t) =  \left( \rho_{0}(\h) + \left(  1 - \rho_{0}(\h) \right)  e^{-t}\right)  {\sinh \h}.
\end{equation}
Note that in \eqref{eqn:J-lambda-mod} we abused notation and wrote (due to symmetry of initial data) $\rho_0(\h)$ for $\rho_0(\bfX)$. The left-hand-side in \eqref{eqn:J-lambda-mod} is a total derivative. Upon integration, we find:
\begin{equation}
\label{eqn:flow-exact-hyper}
\cosh (\lambda(\h,t)) = 1 + \int_0^\h \rho_{0}(\xi) \sinh \xi \, d\xi + e^{-t} \int_0^\h \left(  1 - \rho_{0}(\xi) \right) \sinh \xi \, d\xi,
\end{equation}
where we used the symmetry of the flow, by which the vertex of the hyperboloid ($\h=0$) remains fixed (i.e., $\lambda(0,t)=0$ for all $t$). Equation \eqref{eqn:flow-exact-hyper} provides an exact explicit expression for the flow map $(\h,\phi) \to (\lambda(\h,t),\phi)$.

For the asymptotic behaviour we pass $t \to \infty$ in \eqref{eqn:flow-exact-hyper}. We find:
\begin{equation}
\label{eqn:coshR}
\cosh {R_\h} = 1+ \int_0^\h \rho_{0}(\xi) \sinh \xi \, d\xi.
\end{equation}
where 
\[
R_\h = \lim_{t \to \infty} \lambda(\h,t).
\]
Recall that we have set the total mass to be $1$, which in local coordinates yields:
\begin{equation}
\label{eqn:m1hyperboloid}
1 = \int_0^{2 \pi} \int_0^\pi \rho_{0}(\h) \sinh \h \, d\h d \phi.
\end{equation}

The initial density was assumed of compact support, so consider $\theta_0>0$ such that the support of $\rho_0$ is contained in $0\leq \theta \leq \theta_0$. Then, for any particle trajectory that originates from $\theta>\theta_0$, by \eqref{eqn:coshR} and \eqref{eqn:m1hyperboloid}, we infer that ${R_\h} = R$, where 
\begin{equation}
\label{eqn:R-hyper}
R = \cosh^{-1} \Bigl(1 + \frac{1}{2\pi} \Bigr).
\end{equation}
By elementary geometric considerations \cite{CaFlKePa1997}, $R$ represents the radius of the geodesic disk of unit area centred at the vertex of the hyperboloid. Hence, all trajectories starting from $\theta>\theta_0$ approach the geodesic circle of radius $R$ as $t \to \infty$. 

Finally, by \eqref{eqn:coshR}, $R_\theta$ is monotonic and continuous in $\theta$, so trajectories starting from inside $\h < \h_0$ end up inside the geodesic circle of radius $R$. Combining these facts with the conservation of mass and the fact that $\rho \to 1$ along {\em all} particle paths as $t\to \infty$, we infer that symmetric solutions of model \eqref{eqn:model} approach asymptotically an equilibrium of constant density ($=1$) supported on the geodesic disk of radius $R$ centred at the vertex. That is, initial densities that are symmetric about the vertex, are globally attracted to:
\begin{equation}
\bar{\rho}_H(\bfx)=%
\begin{cases}
1 & \quad\text{ if }\;\h <R \\
0 & \quad\text{ otherwise},
\end{cases}
\label{eqn:sstate-hyper}
\end{equation}
where $\bfx=(\h,\phi)$ and $R$ is given by \eqref{eqn:R-hyper} -- see Figure \ref{fig:hyper-sym} for a numerical illustration using a particle method.
\end{proof}

\begin{rmk}
\label{rmk:symmetry-hyper} Theorem \ref{thm:hyper} holds more generally, for initial densities that are symmetric with respect to an arbitrary point on $H$. Since the interaction potential (and consequently, the interaction forces) only depend on the geodesic distance between points, by symmetry, particles flow along geodesic rays through the symmetry point. The exact expressions for the particle paths, as well as the considerations on the asymptotic behaviour could be then adapted to this more general context.
\end{rmk}


\subsection{Numerical results}
\label{subsubsect:numerics-hyper}

We use the particle method detailed in Section \ref{subsubsect:numerics-sphere}. Specifically, we consider $N$ particles $\bfx_i$ on $H$ and evolve them according to the discrete model
\begin{equation}
\label{eqn:discrete-H}
\frac{d \bfx_i}{dt} = -\frac{1}{N} \sum_{\substack{j = 1 \\ j \ne i}}^N \nabla_{H,i} K(\bfx_i,\bfx_j), \qquad i =1,\dots,N.
\end{equation}
In local coordinates (see the derivation of system \eqref{eqn:ode-sphere} for the sphere), this amounts to solving:
\begin{equation}
\label{eqn:ode-hyper}
\begin{aligned}
\frac{d \hi}{dt} &= -\frac{1}{N}\sum_{\substack{j=1 \\ j \neq i}}^N  \frac{\p K}{\p \hi} (\h_i,\phi_i,\h_j,\phi_j), \qquad i=1,\dots,N\\
\frac{d \phi_i}{dt} &=-\frac{1}{N}\sum_{\substack{j=1 \\ j \neq i}}^N { \frac{1}{\sinh^2{\hi}} \frac{\p K}{\p \phi_i}} (\h_i,\phi_i,\h_j,\phi_j), \qquad i = 1,\dots,N.
\end{aligned}
\end{equation}

Again, the particle method is very appropriate for numerical simulations, given the method of characteristics used in the theoretical results (Section \ref{subsect:asympt-hyper}). The results presented below correspond to numerical solutions of \eqref{eqn:ode-hyper} using the classical Runge-Kutta method. 


\paragraph{Symmetric initial data.} We first validate numerically Theorem \ref{thm:hyper} and consider an initial particle configuration that is symmetric about the vertex (Figure \ref{fig:hyper-sym}(a)). This initial configuration was generated by first placing particles randomly along a given meridian ($\phi=\text{const.}$, $0.2 < \theta < 1.25$), and then rotating the configuration about the $z$-axis. We find that particles evolve along meridians (see the trajectories of the particles indicated by stars) into 
the symmetric configuration shown in Figure \ref{fig:hyper-sym}(b); see also the zoom-in plot in Figure \ref{fig:hyper-sym}(c). The distance from the vertex of the hyperboloid to the particles on the equilibrium's boundary is $0.5742$ which is within $3 \%$ of the geodesic radius computed in \eqref{eqn:R-hyper} ($R \approx 0.5570$). 

We have also performed numerical experiments with larger numbers of particles, to confirm that the particle method gives better numerical approximations of the continuum equilibrium as the number of particles increases. Indeed, with $N=400$ and $N=900$ particles we found symmetric equilibria supported in geodesic disks of radii $0.5692$ and $0.5661$, respectively, which are within $2.1 \%$ and $1.6 \%$  of the continuum geodesic radius in \eqref{eqn:R-hyper}.


\begin{figure}[!htbp]
 \begin{center}
 \begin{tabular}{cc}
 \includegraphics[width=0.5\textwidth]{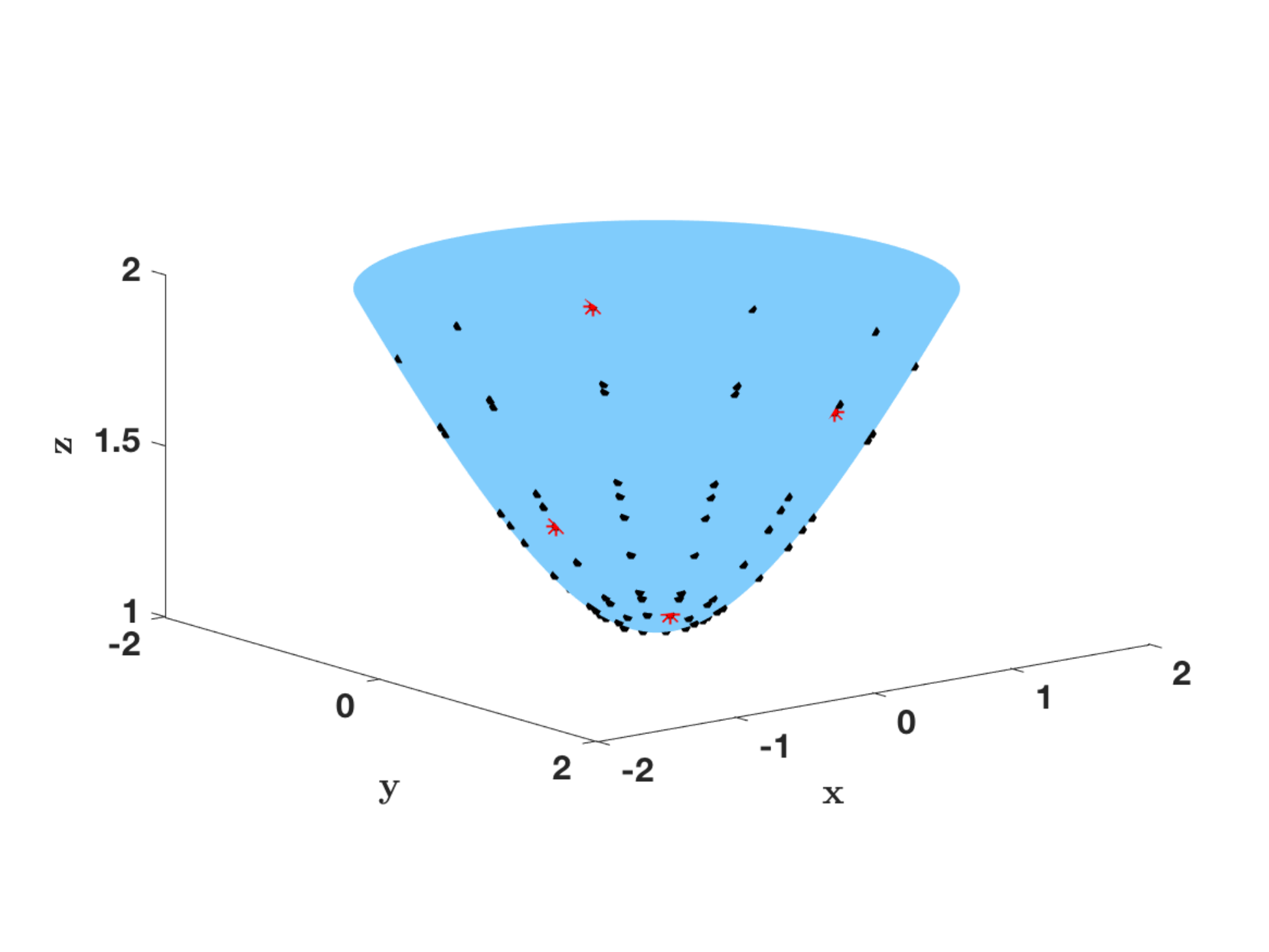}  &
 \includegraphics[width=0.5\textwidth]{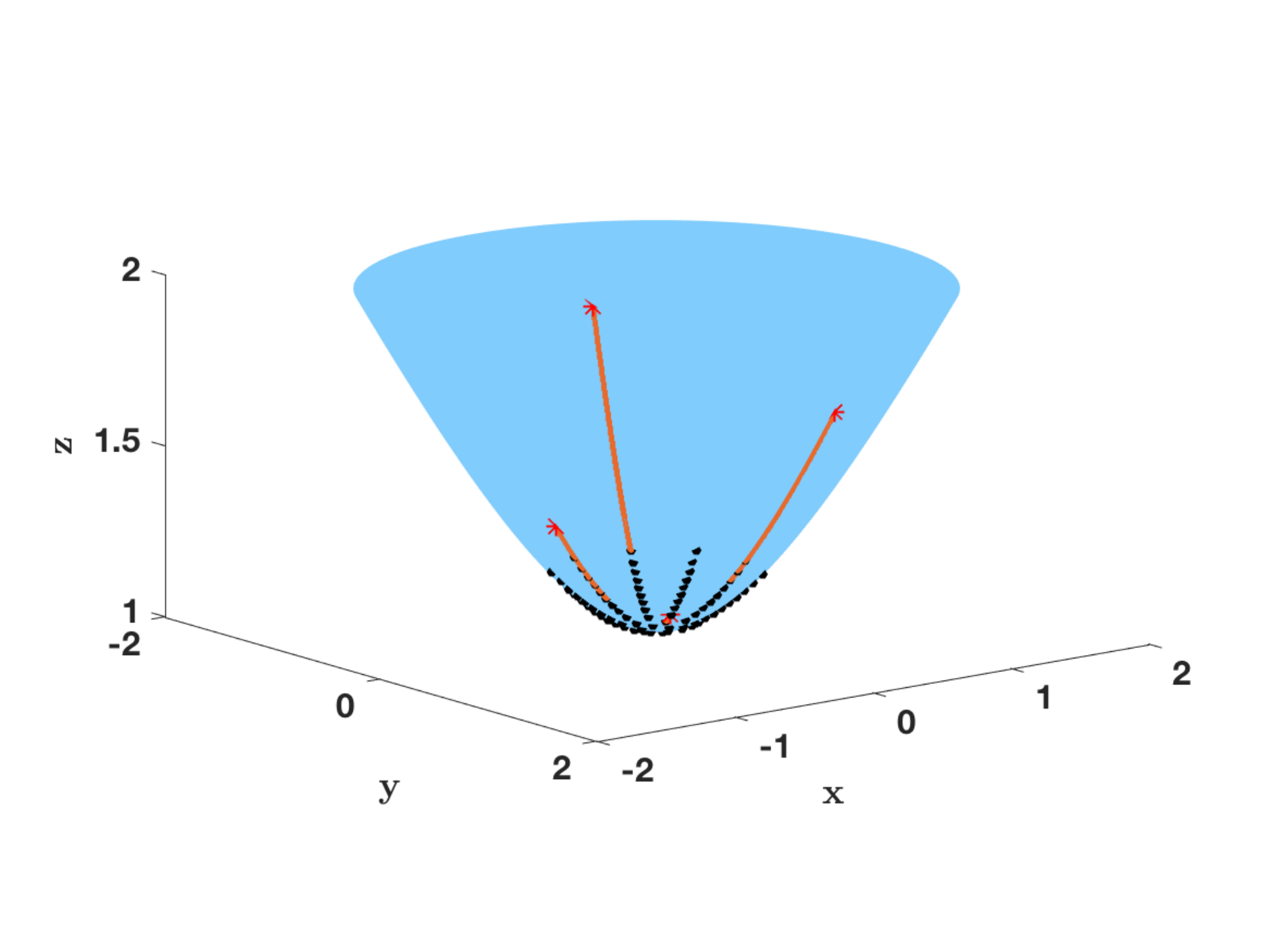} \\
 (a) & (b)
 \end{tabular}
 \begin{center}
 \begin{tabular}{c}
 \includegraphics[width=0.6\textwidth]{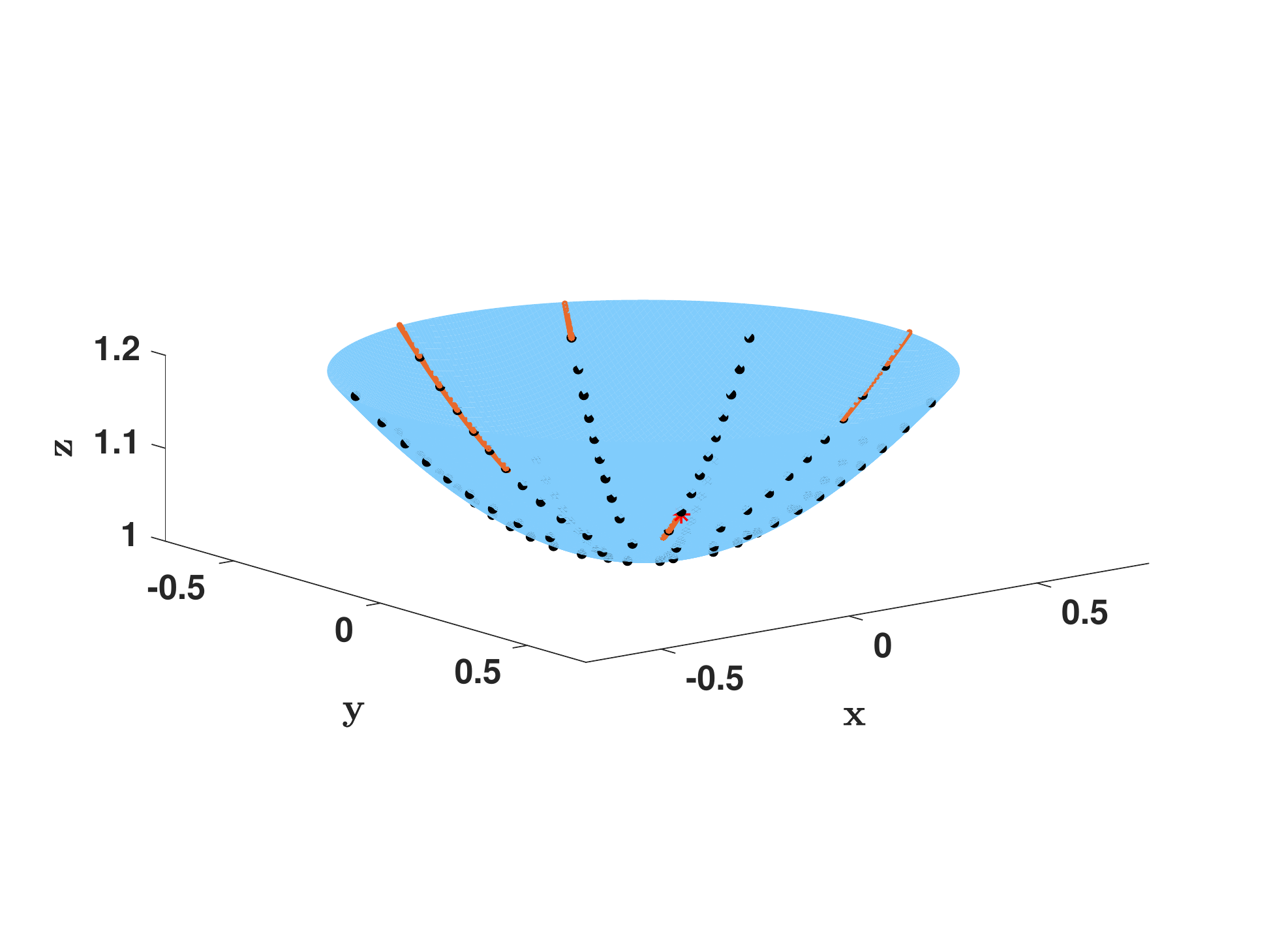}  \\
(c)
 \end{tabular}
 \end{center}
\caption{ Numerical simulation with $N=100$ particles for model \eqref{eqn:model} on $H$ with the interaction potential \eqref{eqn:K-hyper} (see also \eqref{eqn:Green-hyper} and \eqref{eqn:A}). (a) Symmetric (about the vertex) initial configuration on $H$, with $\theta$ coordinates generated randomly in the interval $(0.2,1.25)$.  (b) The configuration remains symmetric for all times and evolves into a uniform (with respect to the metric on $H$) particle distribution supported over a geodesic disk centred at the vertex -- see Theorem \ref{thm:hyper}. The solid lines represent the trajectories of the particles indicated by stars in figure (a). (c) Zoom-in of figure (b) on the equilibrium configuration.}
\label{fig:hyper-sym}
\end{center}
\end{figure}


\paragraph{Arbitrary initial data.} Our numerical simulations have indicated that all (not just the symmetric) solutions to model \eqref{eqn:model} on $H$, with interaction potential given by \eqref{eqn:K-hyper}, approach asymptotically constant equilibrium densities supported on geodesic disks of radius $R$ given by \eqref{eqn:R-hyper}. For the model in Euclidean space the analogous result was illustrated and proved in \cite{FeHuKo11, BertozziLaurentLeger}. To have such a result hold for the model on the hyperbolic plane suggests that the construction 
of the interaction potential in \eqref{eqn:K-hyper} can be potentially extended, with similar outcomes, to other non-compact manifolds.

Figure \ref{fig:hyper} corresponds to a numerical simulation with $N=100$ particles initiated at a randomly-generated configuration on $H$, with the $\theta$ and $\phi$ coordinates of the particles drawn randomly in the intervals $(0.3,2.3)$ and $(0,\pi/2)$, respectively. Figure \ref{fig:hyper}(a) shows the initial configuration and Figure \ref{fig:hyper}(b) shows the equilibrium state, along with some particle trajectories (see also the zoom-in plot in Figure \ref{fig:hyper}(c)). The equilibrium density is constant within its support (see \eqref{eqn:rho-flow-hyper} and \eqref{eqn:rho-h}), which for a particle simulation translates into a uniform distribution (with respect to the metric on $H$). We found very similar results with larger number of particles and different initialization procedures. Equilibria obtained from simulations with $N=400$ and $N=900$ will be discussed below.


Next we provide evidence that the support of the equilibrium is a geodesic disk. We illustrate the procedure for the equilibrium in Figure \ref{fig:hyper}(c). To check for the equilibrium's shape, we calculate numerically the Riemannian centre of mass of the equilibrium configuration. The Riemannian centre of mass is guaranteed to exist in such case, given that $H$ has negative curvature everywhere \cite{Afsari2011}. To locate the centre of mass we use the intrinsic gradient descent algorithm investigated in \cite{AfsariTronVidal2013}; recall that the Riemannian centre of mass of a set of points on a manifold minimizes the sum of squares of the geodesic distances to the data points, so using a gradient decent method is very natural here. Figure \ref{fig:Ri}(a) shows the centre of mass $C$ (red diamond) of the equilibrium configuration in Figure \ref{fig:hyper}(c) computed with this method.

After we locate the particles on the equilibrium's boundary (see filled blue circles in Figure \ref{fig:Ri}(a)) we compute the geodesic distances between the Riemannian centre of mass and the boundary points. We denote these distances by $R_i$ (note that the particles on the boundary have been relabelled so that they have consecutive indices starting from $1$) -- see Figure \ref{fig:Ri}(a) for an illustration. Figure \ref{fig:Ri}(b) shows these distances for three simulations.  For the simulation with $N=100$ discussed above, there are $31$ particles on the boundary.  Their distances to the Riemannian centre of mass of the equilibrium are shown in the figure as magenta circles connected by dotted lines, where the thick dotted line represents their mean value. The mean value is $0.5109$ with a relative standard deviation of $0.55\%$. The blue circles connected by dash-dotted lines correspond to a simulation with $N=400$ particles ($67$ particles on the equilibrium's boundary). The mean value (thick dash-dotted line) is $0.5330$ with a relative standard deviation of $0.47\%$. Finally, the red circles connected by dashed lines represent the distances from the centre to the boundary points for a simulation with $N=900$ particles ($98$ of which on the boundary). Their mean value (thick dashed line) is $0.5415$ with a relative standard deviation of $0.20\%$.

The results presented in Figure \ref{fig:Ri} strongly suggest that the continuum equilibrium is supported on a geodesic disk of radius $R$  given by \eqref{eqn:R-hyper} (this value has been indicated as a thick black solid line in the figure; recall $R \approx 0.5570$). Indeed, as the number of particles increases, not only that the mean value of the distances $R_i$ approach $R$, but also their relative standard deviation decreases. In other words, the larger the number of particles, the closer the boundary of the particle equilibrium is to a geodesic circle of radius $R$. Similar results were obtained with a variety of initial configurations and different number of particles. 

We conclude this section by posing the following conjecture: Geodesic disks of constant density are global attractors for model \eqref{eqn:model} on $H$ with the interaction potential \eqref{eqn:K-hyper}. Recall that the analogous statement holds for the model in Euclidean space with potential \eqref{eqn:K-plane} \cite{FeHuKo11, BertozziLaurentLeger}. One major obstacle for proving such a global convergence result is that the aggregation model on manifolds (and in particular on $H$) does not necessarily conserve the Riemannian centre of mass (see Section \ref{sect:prelim}). In Euclidean spaces the centre of mass is conserved and the general proof in \cite{BertozziLaurentLeger} relies fundamentally on this fact. Consequently, while the centre of the attracting geodesic disk is a priori known in Euclidean spaces (it is the centre of mass of the initial density), it is not known for the aggregation model on $H$. For this reason the proof for global attractors in the plane does not immediately extend to $H$.

\begin{figure}[!htbp]
 \begin{center}
 \begin{tabular}{cc}
 \includegraphics[width=0.5\textwidth]{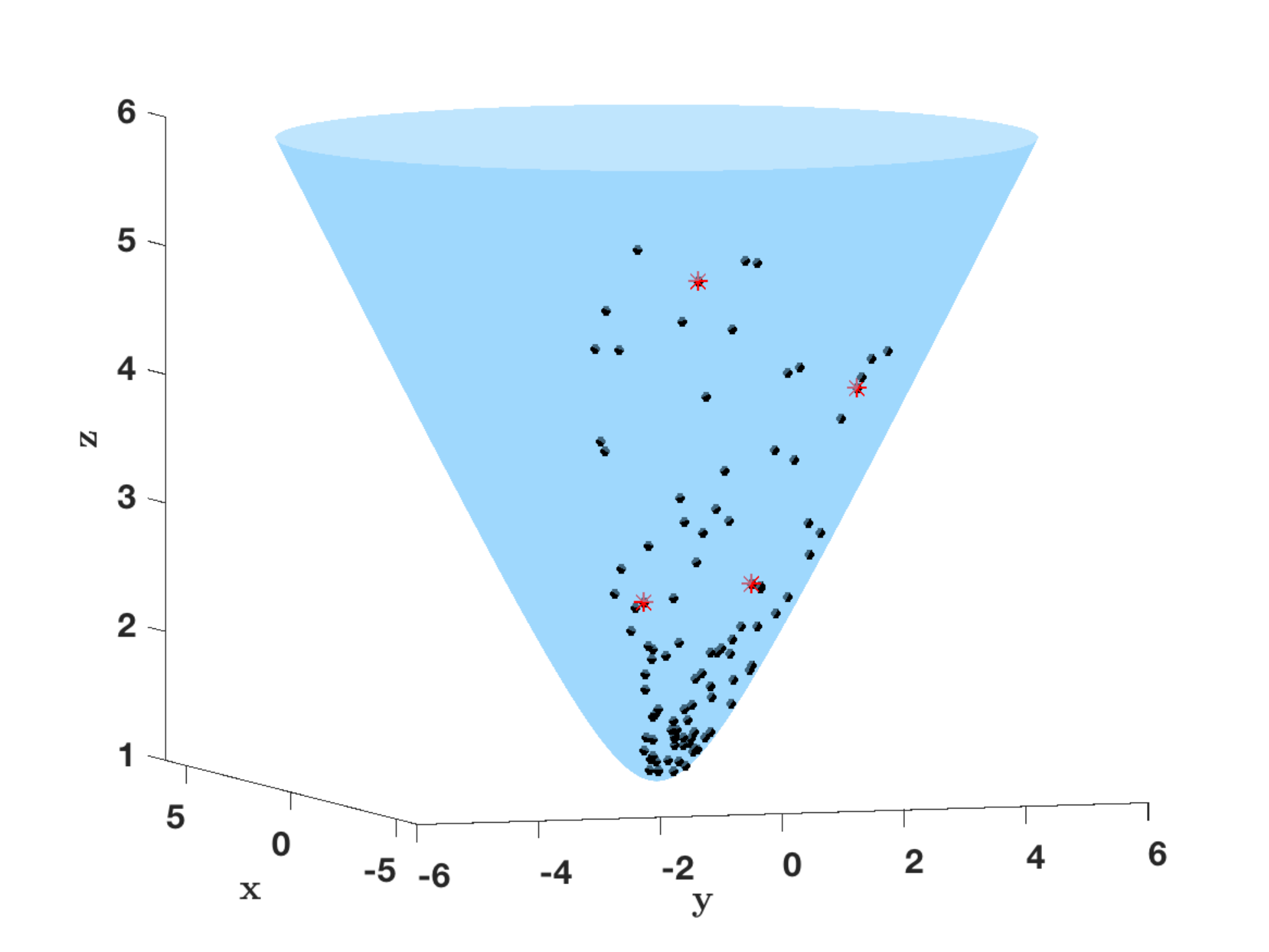} &  \includegraphics[width=0.5\textwidth]{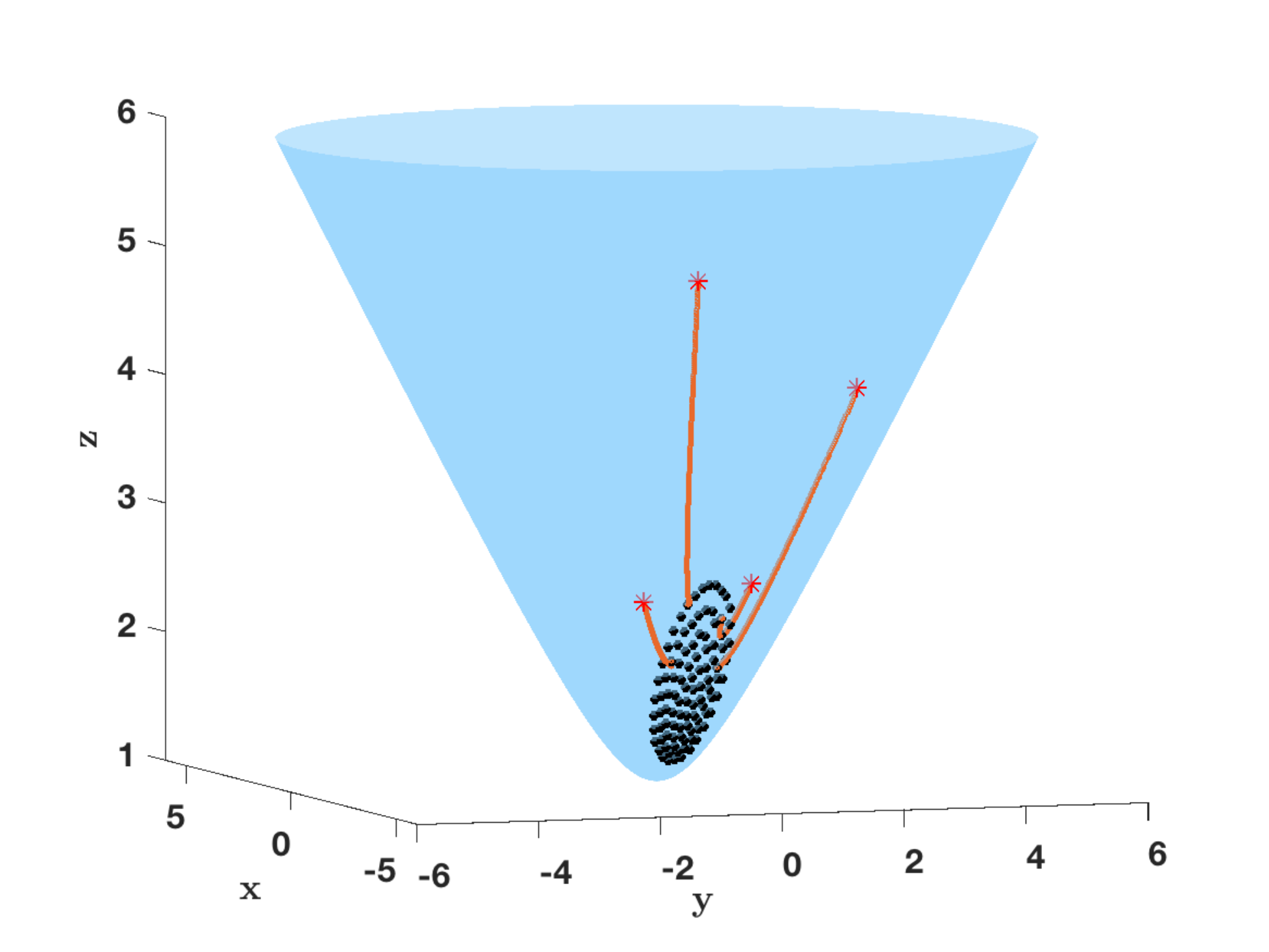} \\
 (a) & (b)
 \end{tabular}
\begin{center}
\begin{tabular}{c}
\includegraphics[width=0.5\textwidth]{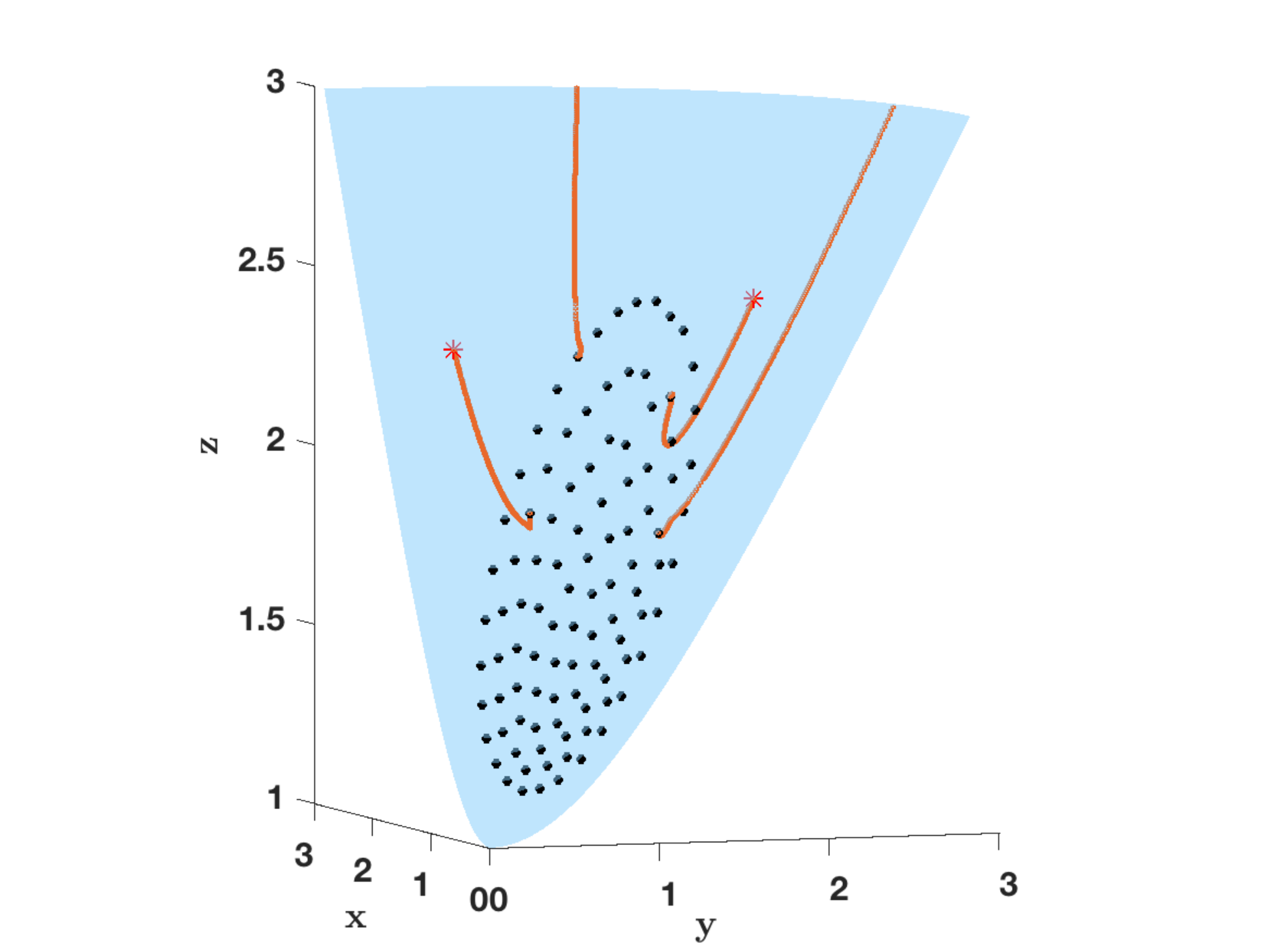}  \\
(c) 
\end{tabular}
\end{center}
\caption{Numerical simulation with $N=100$ particles for model \eqref{eqn:model} on $H$ with the interaction potential \eqref{eqn:K-hyper} (see also \eqref{eqn:Green-hyper} and \eqref{eqn:A}). (a) Random initial configuration on $H$, with coordinates $\theta$ and $\phi$ drawn randomly from intervals $(0.3,2.3)$ and $(0,\pi/2)$, respectively.  (b) Equilibrium state corresponding to the initial configuration in (a). The solid lines represent the trajectories of the particles indicated by stars. (c) Zoom-in of the equilibrium state in figure (b). Numerical investigations (see Figure \ref{fig:Ri}) suggest that the equilibrium configuration consists of a uniform (with respect to the metric on $H$) particle distribution supported over a geodesic disk of radius $R$ (see \eqref{eqn:R-hyper}).}
\label{fig:hyper}
\end{center}
\end{figure}



\begin{figure}[!htbp]
\begin{center}
\begin{tabular}{cc}
 \includegraphics[width=0.5\textwidth]{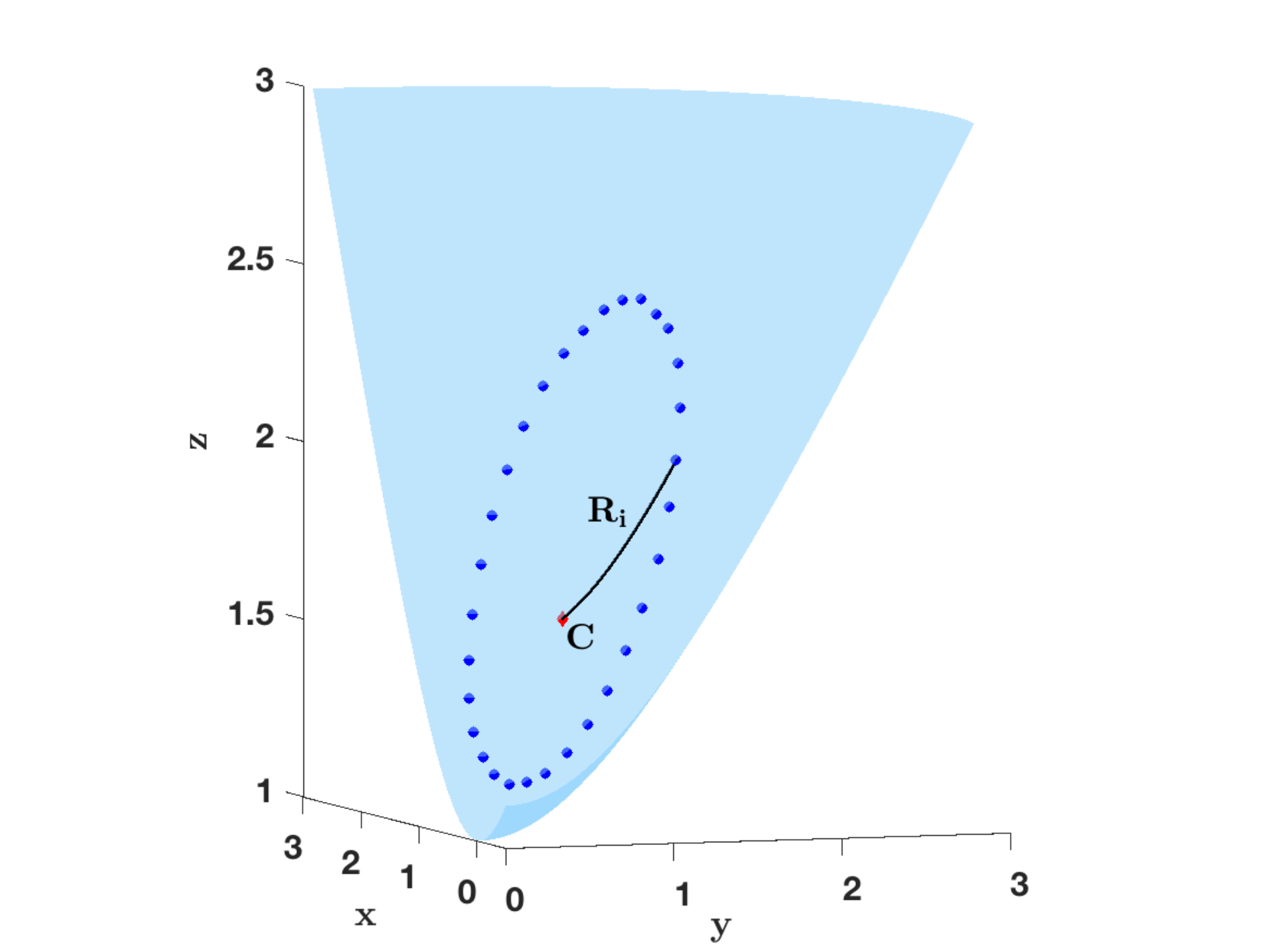} & \includegraphics[width=0.5\textwidth]{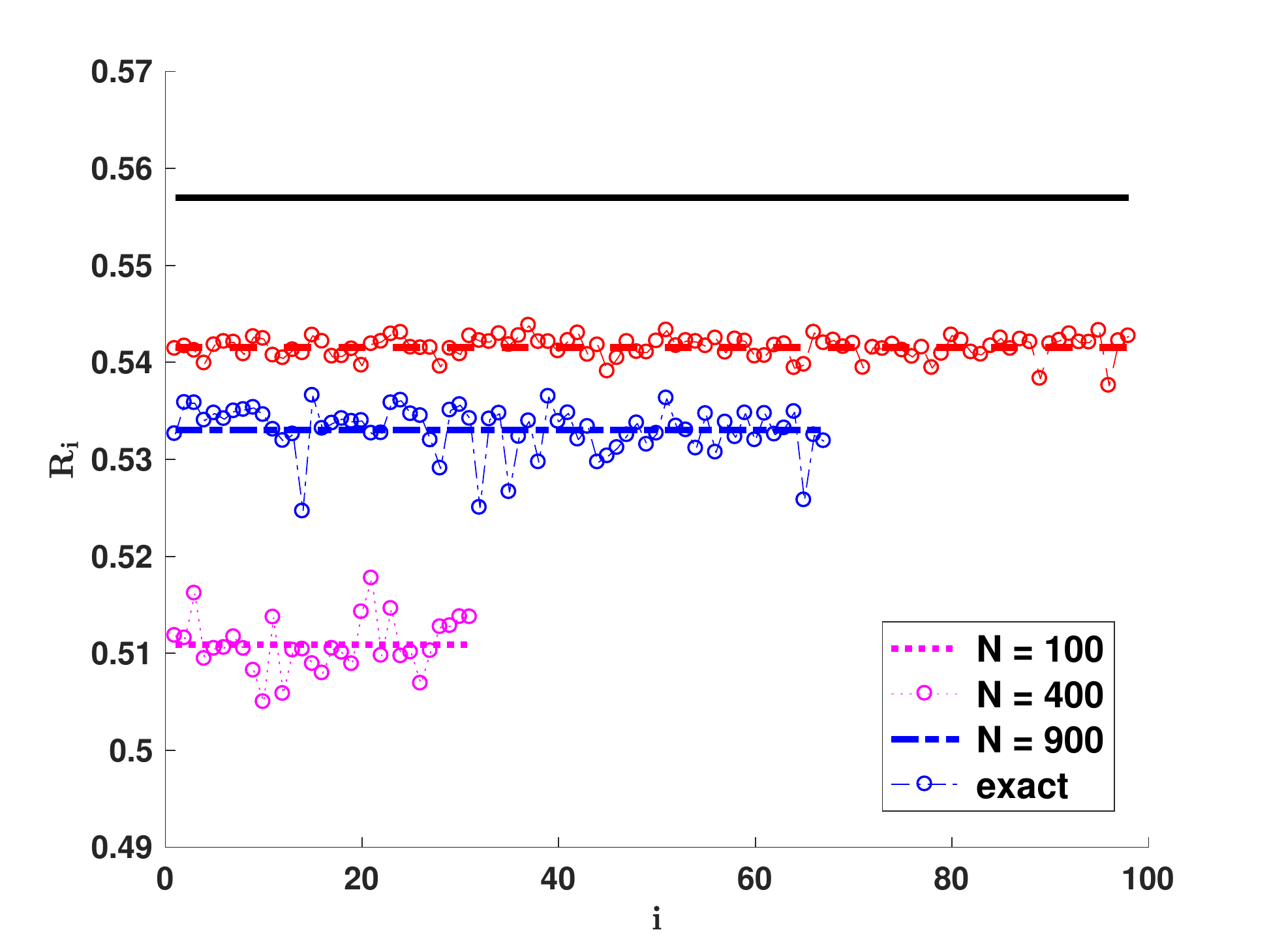} \\
 (a) & (b) 
 \end{tabular}
\caption{Numerical investigation of the equilibrium configurations on $H$. (a) Boundary points (filled blue circles) along with the Riemannian centre of mass (red diamond) of the equilibrium state in Figure \ref{fig:hyper}(c). (b) Distances from the Riemannian centre of mass to the particles on the equilibrium's support for 3 simulations: $N=100$, $400$ and $900$. Particles on the boundary have been relabelled to have consecutive indices starting from $1$. The distances are shown in circles (connected by dotted, dash-dotted and dashed lines, respectively) and the corresponding thick lines represent their mean values. There are $31$, $67$ and $98$ boundary points for the 3 simulations, with mean distances to centre of $0.5109$, $0.5330$ and $0.5415$, and relative standard deviations of $0.55\%$, $0.47\%$ and $0.20\%$, respectively. The results strongly suggest that the continuum equilibrium is supported on a geodesic disk of radius $R$  given by \eqref{eqn:R-hyper} (this value has been indicated as a thick solid line; $R \approx 0.5570$). }
\label{fig:Ri}
\end{center}
\end{figure}


\section{Other potentials and closing remarks}
\label{sect:future}

Potentials in power-law form have been frequently considered in the literature on the aggregation model \eqref{eqn:model}  in Euclidean spaces \cite{BaCaLaRa2013, FeHu13, FeHuKo11, KoSuUmBe2011, Brecht_etal2011}.  In our context, by using the geodesic distance between points, a power-law potential reads:
\begin{equation}
\label{eqn:K-pl}
K(\bfx,\bfx') = -\frac{1}{p}d(\bfx,\bfx')^p + \frac{1}{q}d(\bfx,\bfx')^q,
\end{equation}
where the exponents $p$ and $q$ (with $p<q$) correspond to repulsive and attractive interactions, respectively.

As shown in various works \cite{FeHu13, KoSuUmBe2011, Brecht_etal2011}, the delicate balance between attraction and repulsion often leads to complex equilibrium configurations, supported on sets of various dimensions. The existence and characterization of equilibria as minimizers of the interaction energy has been an active research topic lately (existence was investigated
in \cite{ChFeTo2015, SiSlTo2015, CaCaPa2015}, the dimensionality of local minimizers was studied in \cite{Balague_etalARMA}). Below, we consider the aggregation model \eqref{eqn:model} on the hyperbolic plane $H$ and demonstrate numerically that a variety of equilibria can also be obtained with a power-law potential on such manifold.

Figure \ref{fig:hyper-pl}(a)-(c) shows three equilibria obtained numerically for model \eqref{eqn:model} on $H$ with interaction potential \eqref{eqn:K-pl}. The equilibria are of very different type and dimension: (a) an equilibrium supported on an annular region, (b) concentration on a geodesic circle (ring), and (c) delta accumulation on three points. We note that similar equilibria have been found for the aggregation model in Euclidean plane with power-law interaction potentials \cite{FeHu13, KoSuUmBe2011, ChFeTo2015, Balague_etalARMA}. We also point out that to establish whether certain boundaries/shapes consist of geodesic circles, we followed a very similar procedure to that discussed in the previous section (see Figure \ref{fig:Ri}). Namely, we located the Riemannian centre of mass of the equilibrium, along with the particles on the boundaries, and computed the mean distance (radius) from the centre of mass to these particles. In all cases, including the simulations presented below, the radii of these various circles had a relative standard deviation within 1\% (in most cases much smaller in fact).


Another class of interaction potentials which has been widely used in the literature on model \eqref{eqn:model} in $\R^n$ consists of (generalized) Morse-type potentials  \cite{CaHuMa2014}: 
\begin{equation}
\label{eqn:K-gM}
K(\bfx,\bfx')  = V(d(\bfx,\bfx')) - CV (d(\bfx,\bfx')/l),
\end{equation}
where
\begin{equation}
\label{eqn:Morse}
V(r) = -e^{-\frac{r^s}{s}}, \qquad \text{ with } s > 0.
\end{equation}
Here, $C$ and $l$ are positive constants, which control the relative size and range of the repulsive interactions. In one dimension with $s=1$, $V(r)$ is a multiple of the Green's function of the differential operator $\p^2_r - \operatorname{Id}$. This property enables explicit calculations of the equilibrium solutions by converting integral equations for equilibria into differential equations  \cite{BeTo2011}. Potentials of form \eqref{eqn:K-gM}-\eqref{eqn:Morse} have been also used in other models for swarming and 
flocking \cite{Chuang_etal}.

The aggregation model \eqref{eqn:model} in the Euclidean plane with Morse-type interaction potentials exhibits a wide variety of possible equilibria \cite{CaHuMa2014}. In Figure \ref{fig:hyper-pl}(d)-(f) we showcase some equilibria obtained numerically for the model in the hyperbolic plane with potential \eqref{eqn:K-gM}-\eqref{eqn:Morse}. We picked values of the parameters $C$, $l$ and $s$ that have been used for the Euclidean model \cite{CaHuMa2014}. In each plot we observe mixed dimensionality of the equilibrium's support: (d) geodesic circle (ring) with a delta accumulation at the centre, (e) concentration on a ring with a continuous density supported on a concentric geodesic disk, (f) concentration on a ring with a continuous density inside. 

\begin{figure}[!htbp]
 \begin{center}
 \begin{tabular}{ccc}
 \includegraphics[width=0.33\textwidth]{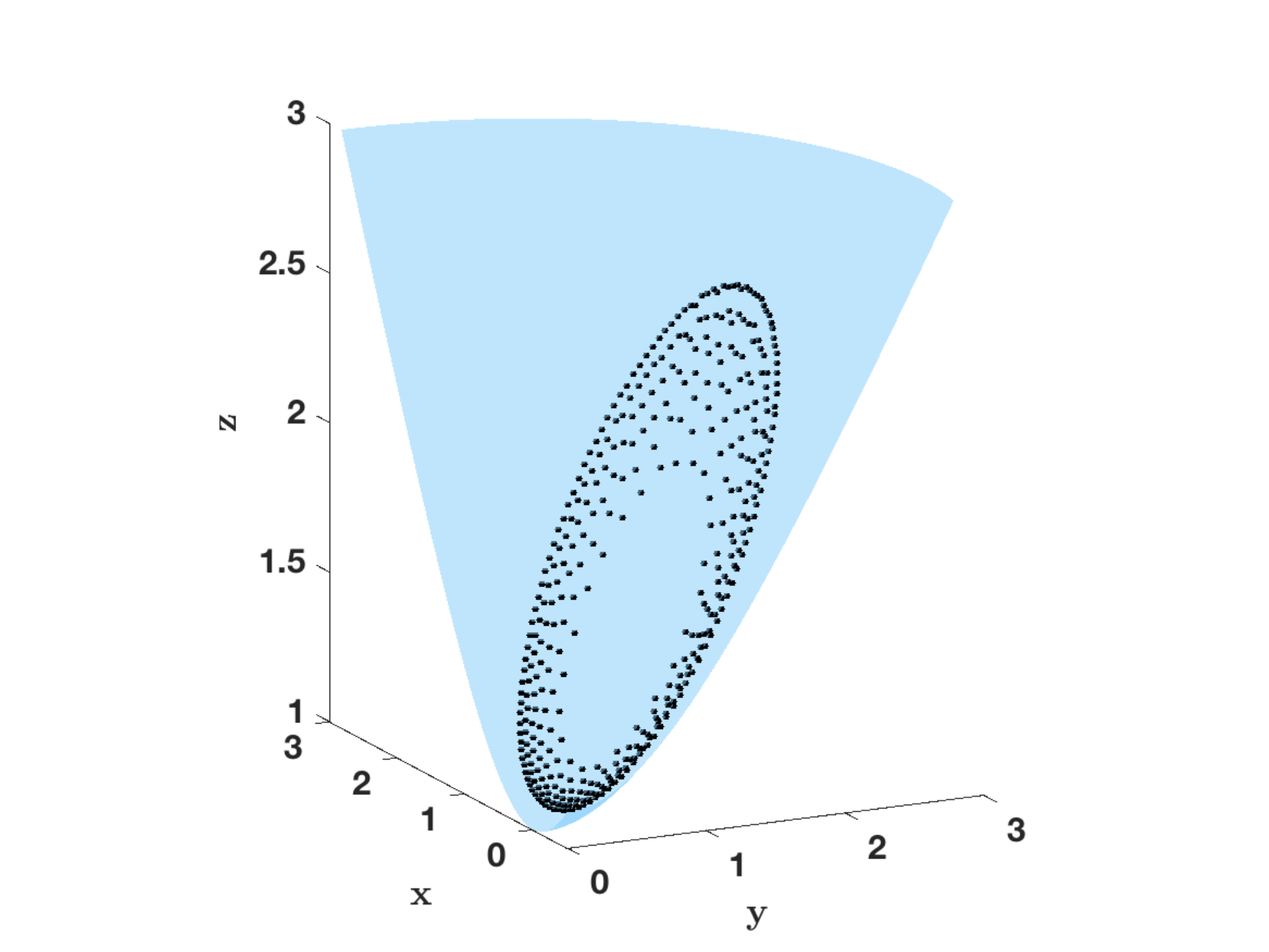} & 
 \includegraphics[width=0.33\textwidth]{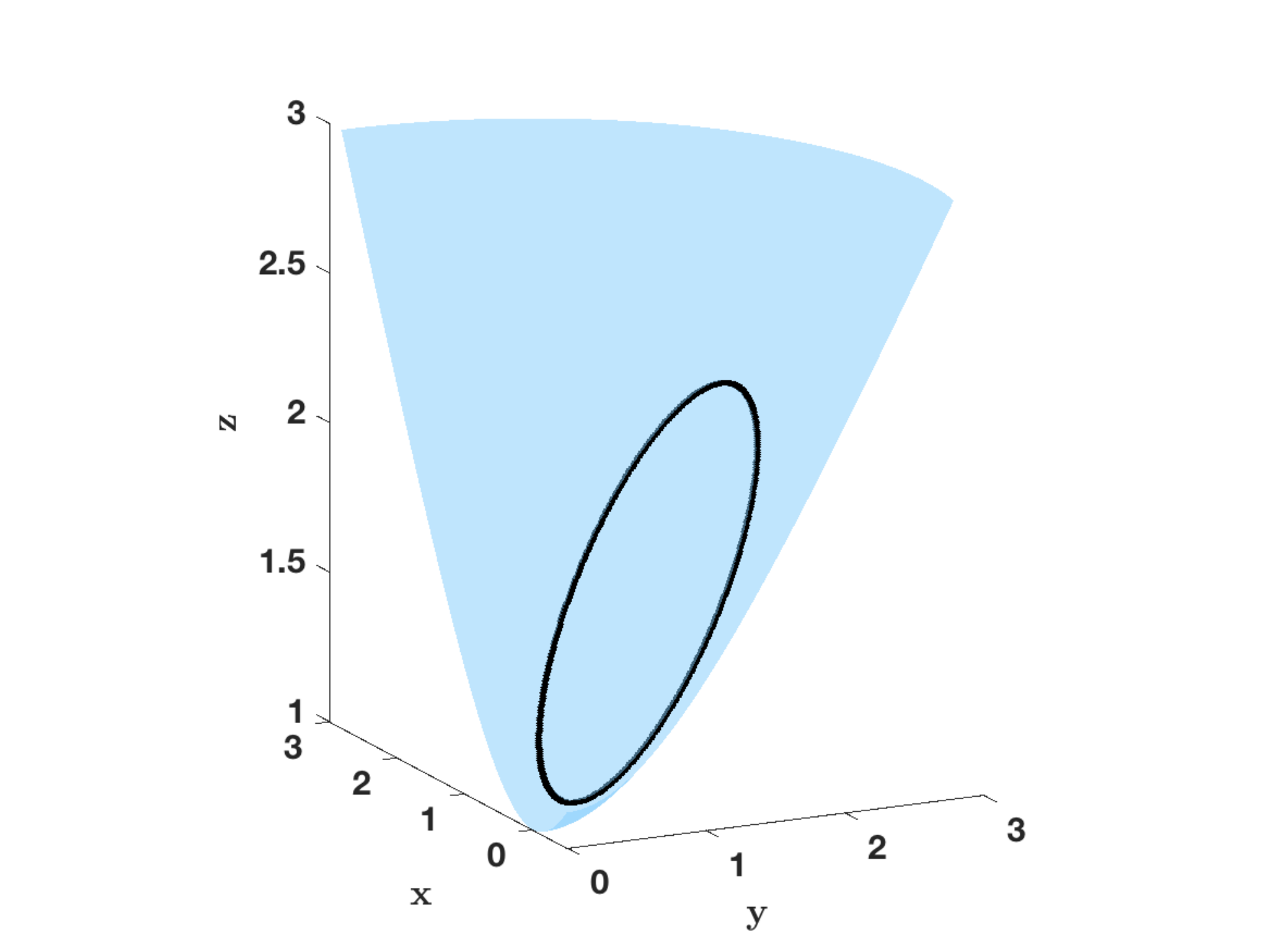}  &
 \includegraphics[width=0.33\textwidth]{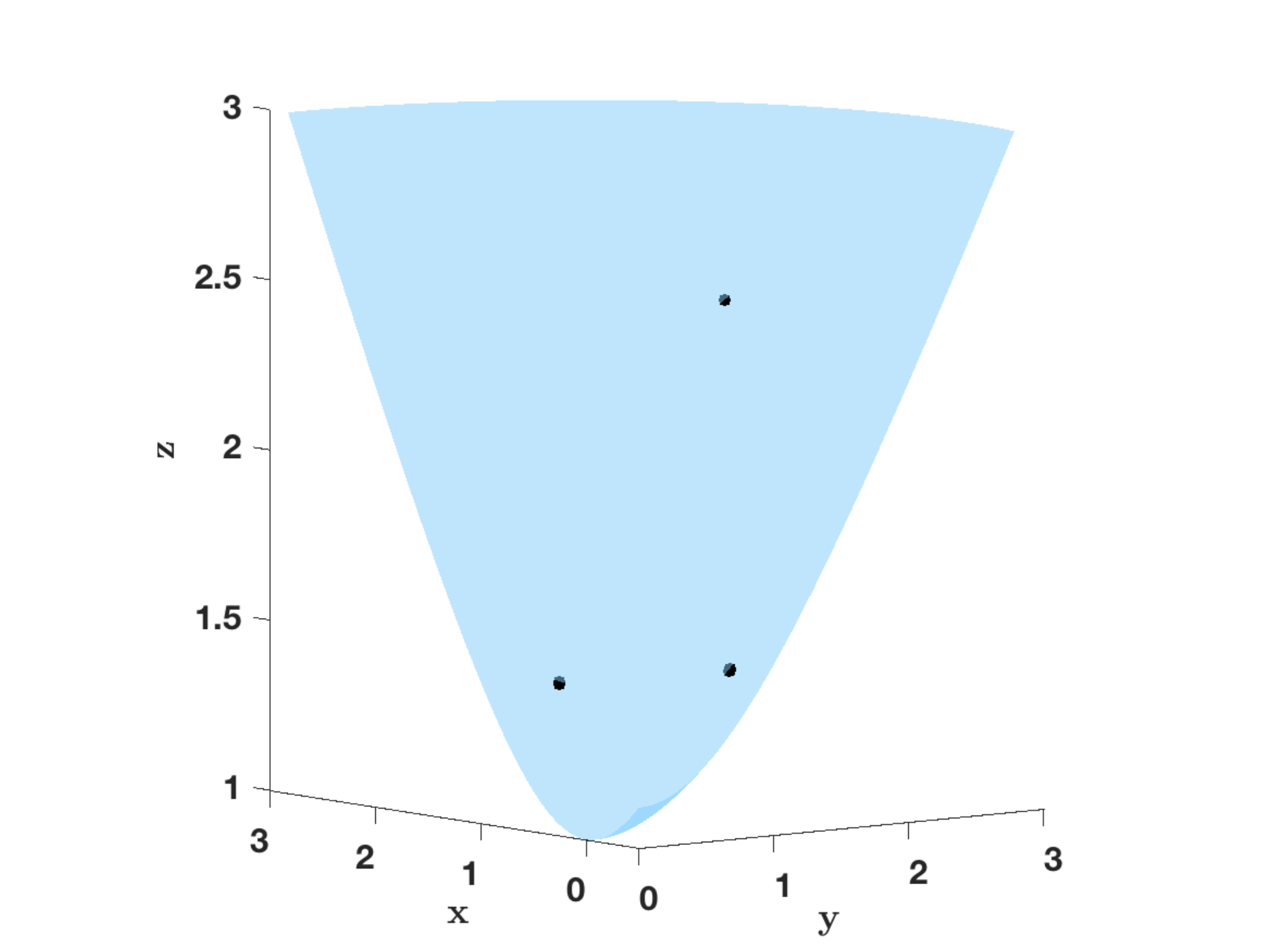} \\
 (a) & (b) & (c) \\
 \includegraphics[width=0.33\textwidth]{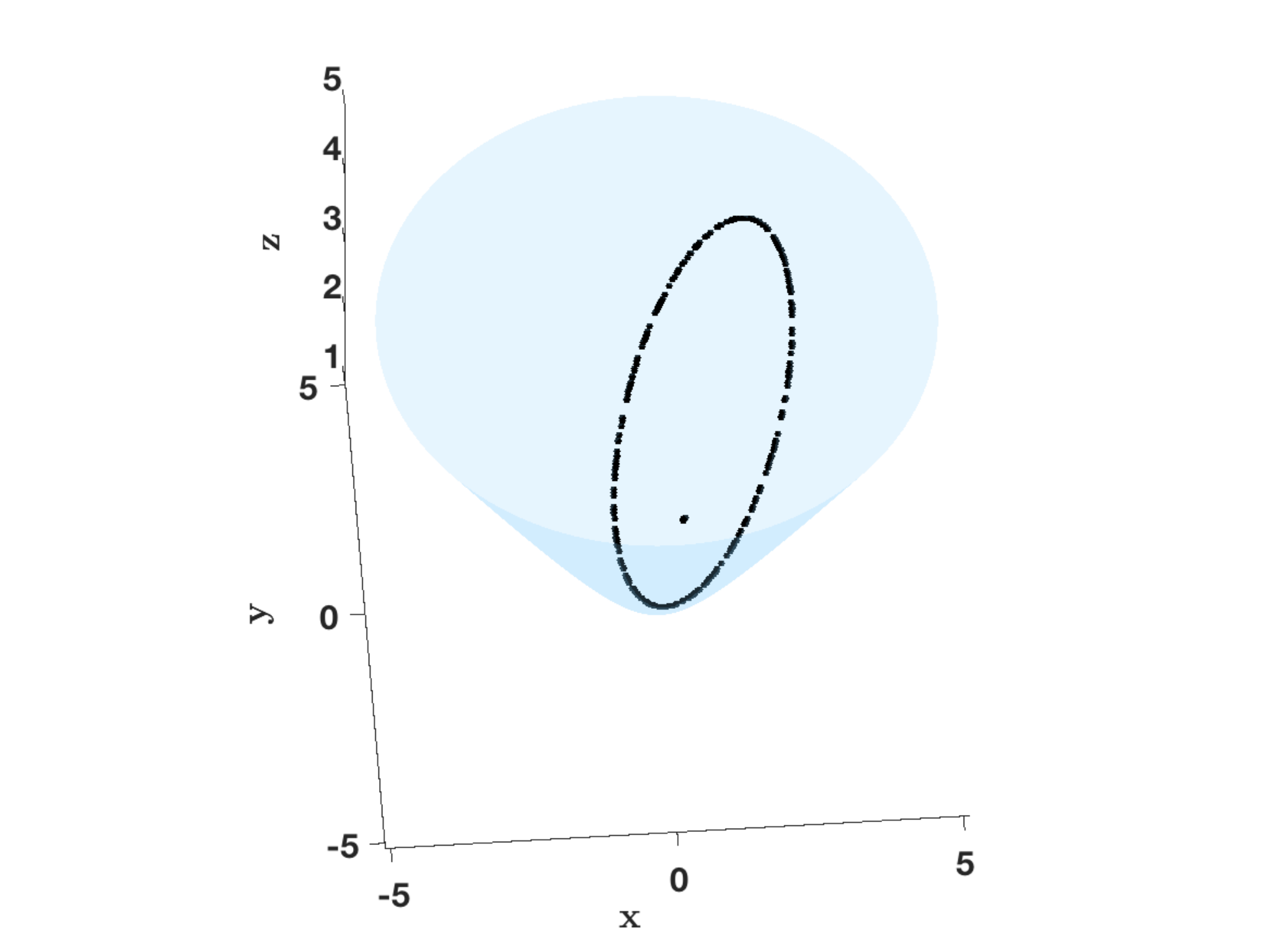} &
 \includegraphics[width=0.33\textwidth]{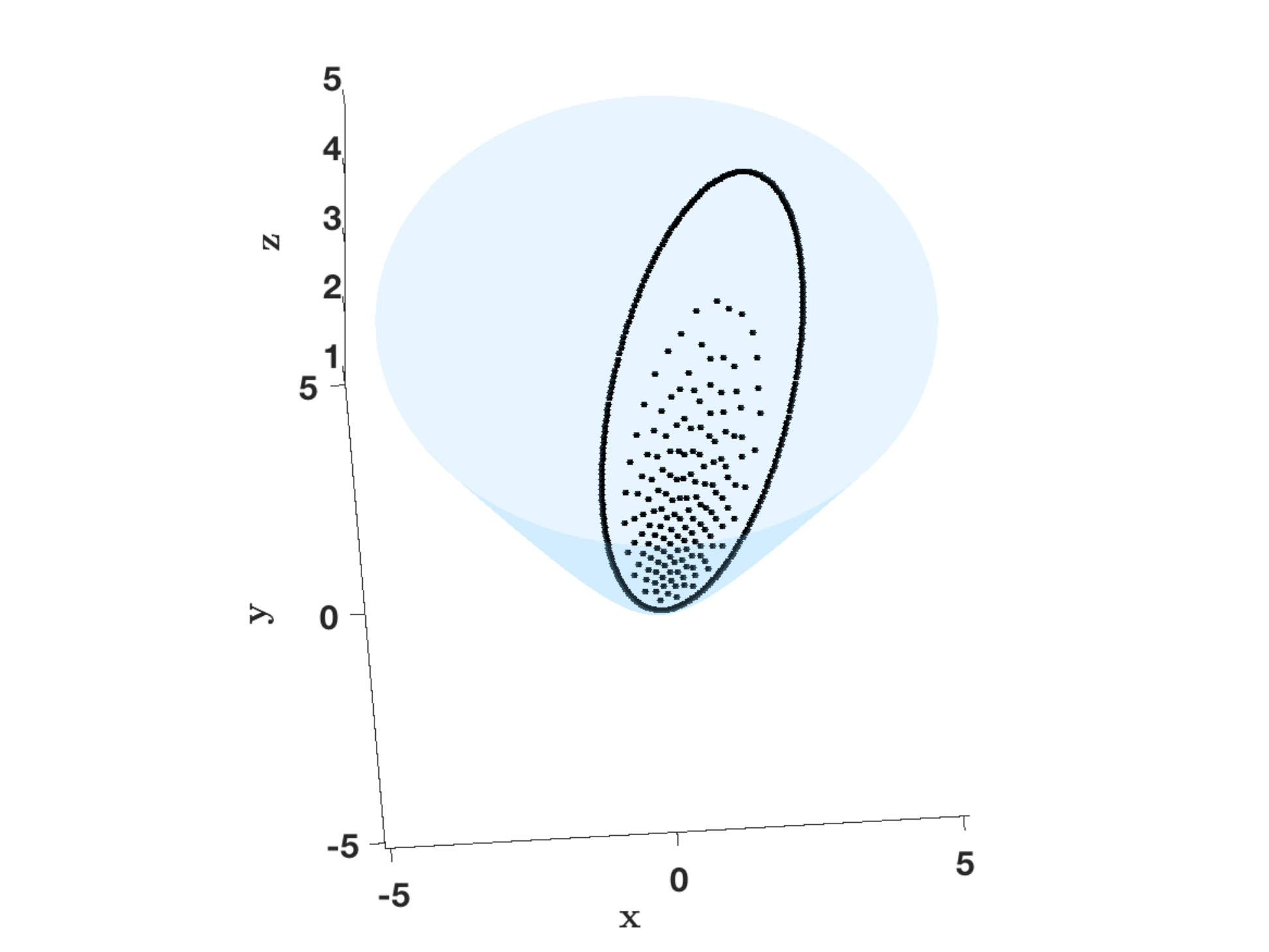} &
 \includegraphics[width=0.33\textwidth]{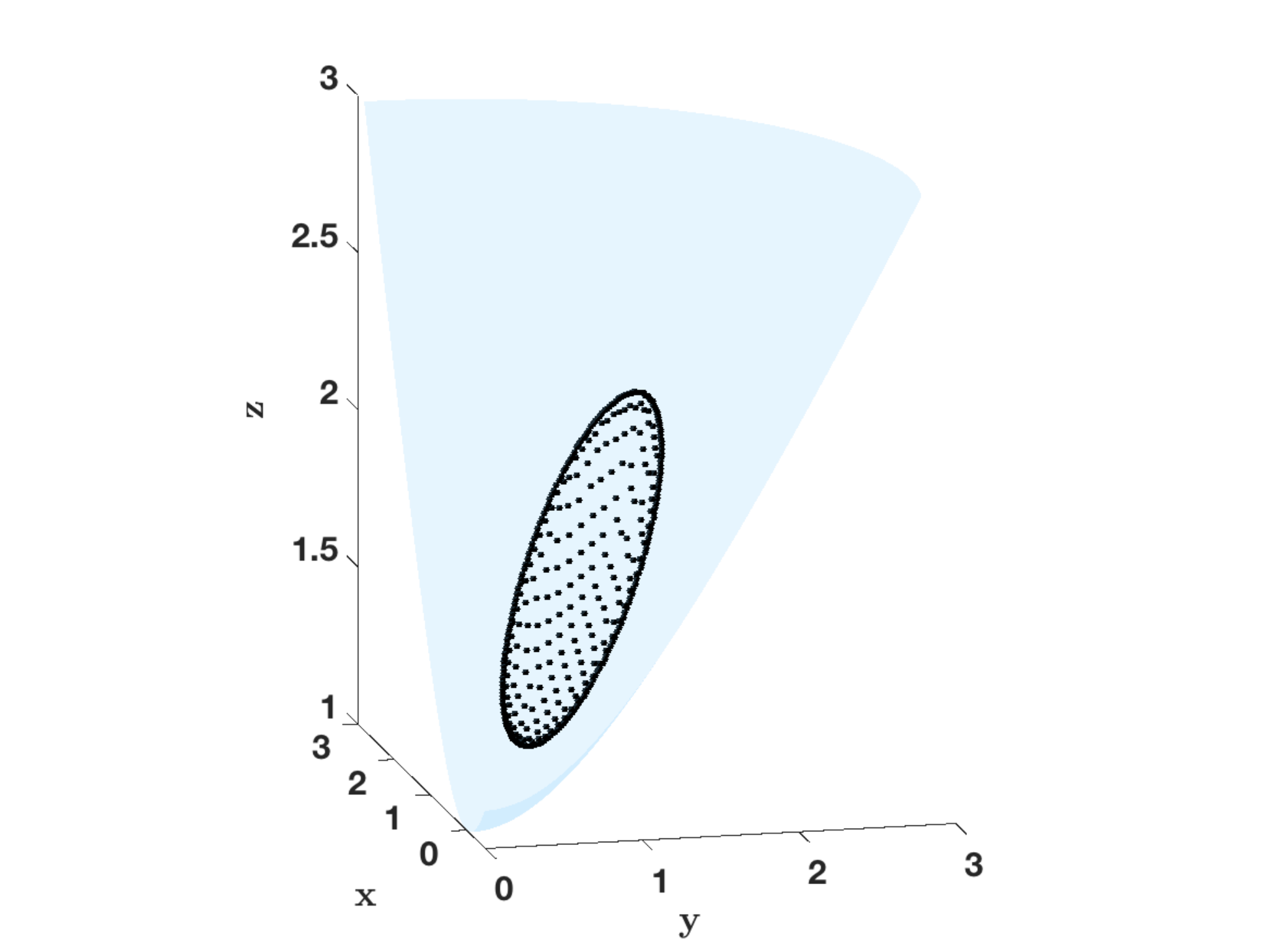} \\
  (d) & (e) & (f) 
 \end{tabular}
 \begin{center}
 \end{center}
\caption{Numerical explorations of equilibria of model \eqref{eqn:model} on $H$ with the power-law interaction potential \eqref{eqn:K-pl} (plots (a)-(c))) and the Morse-type potential \eqref{eqn:K-gM}-\eqref{eqn:Morse} (plots (d)-(f)). (a) $p=0.5$, $q=5$. Equilibrium density supported on an annular region. (b) $p=1$, $q=8$. Concentration on a geodesic circle (ring). (c) $p=6$, $q=7.5$. Equilibrium consists of a delta accumulation on three points. (d) $C = 1.2$, $l = 0.75$, $s = 2$. Equilibrium supported on a geodesic circle with a delta accumulation at the centre.  (e) $C = 1.2$, $l = 0.75$, $s = 1.8$. Concentration on a ring with a continuous density supported on a concentric interior disk. (f) $C = 0.6$, $l = 0.5$, $s = 1.5$. Concentration on a ring with a continuous density inside.}
\label{fig:hyper-pl}
\end{center}
\end{figure}

A rigorous investigation of equilibria shown in Figure \ref{fig:hyper-pl} is challenging. One possible direction could be to study the stability of some of these equilibria (e.g., the ring, the annular region), extending similar analyses in Euclidean spaces \cite{KoSuUmBe2011, Brecht_etal2011, BaCaLaRa2013}. Also, the results in Figure \ref{fig:hyper-pl} suggest that the dimensionality of the equilibria on $H$ relates to the strength of repulsion (value of the exponent $p$ for potential \eqref{eqn:K-pl}), as for the model in the Euclidean plane (see \cite{Balague_etalARMA}). Investigating this connection further is also an interesting future direction that can be pursued.


In closing, we believe that the aggregation model on general manifolds that has been proposed in this paper, the general construction of an interaction potential that leads to constant density equilibria on the sphere and the hyperbolic plane (along with the analytical considerations that can be made in such cases), as well as the various numerical illustrations that demonstrated swarming with other interaction potentials, will have set up a framework for and motivate further research and developments on self-organization models and their applications.

\bibliographystyle{abbrv}
\def\url#1{}

\def\cprime{$'$}


\end{document}